\documentclass[twocolumn]{autart}    

\usepackage{graphicx}          

\usepackage{cite}
\usepackage[colon]{natbib}
\usepackage{amsmath,amssymb,amsfonts}
\usepackage{threeparttable}
\usepackage{color}
\usepackage{xcolor}
\usepackage{makecell}
\usepackage{algorithmic}
\usepackage{algorithm,algorithmic}
\newtheorem{asum}{Assumption}
\newcommand*{\circled}[1]{\lower.7ex\hbox{\tikz\draw (0pt, 0pt)%
		circle (.42em) node {\makebox[1em][c]{\small #1}};}}
\usepackage{textcomp}

\begin{document}

\begin{frontmatter}

\title{Gradient-extrapolation-based distributed mirror descent algorithm for multi-cluster aggregative games} 

\thanks[footnoteinfo]{This work was supported by the National Natural Science Foundation of China under Grant 62103203. (Corresponding author: Fuyong Wang.)}

\author[Paestum]{Rui Zhu}\ead{rui$\_$zcg@mail.nankai.edu.cn},    
\author[Paestum]{Fuyong Wang}\ead{wangfy@nankai.edu.cn},               
\author[Paestum]{Zhongxin Liu}\ead{lzhx@nankai.edu.cn},  
\author[Paestum]{Zengqiang Chen}\ead{chenzq@nankai.edu.cn}

\address[Paestum]{College of Artificial Intelligence, Nankai University, Tianjin 300350, China}  


\begin{abstract}                          
This paper studies a class of multi-cluster aggregative games characterized by the coexistence of cooperation and competition, where each agent's cost function depends on its own strategy and the aggregate of all agents' strategies. To address the Nash equilibrium seeking problem for such games in the non-Euclidean setting, a distributed mirror descent algorithm with gradient extrapolation is proposed over time-varying intra-cluster and inter-cluster networks. The mirror descent framework employs a general Bregman divergence as the distance measure, providing greater flexibility than Euclidean-based methods, while gradient extrapolation exploits historical gradient information to improve convergence performance. Under the restricted strong monotonicity characterized by the Bregman divergence, the convergence of the proposed algorithm is established, and it achieves the $\mathcal{O}(1/k)$ convergence rate with the appropriately selected step-size and parameters. Finally, the effectiveness of the proposed algorithm is verified by an example on the demand response of energy systems.
\end{abstract}

\end{frontmatter}

\thispagestyle{empty}

\section{Introduction}
\label{sec:introduction}
Aggregative games, as an important subclass of non-cooperative games, have been extensively utilized to address optimal decision-making problems in multi-agent systems, such as the demand response management \citep{Ye17}, charging coordination in electric vehicles \citep{Zhou25}, and flow control in traffic networks \citep{Bar15}, among others. Distinct from general non-cooperative games, the cost of each agent in such games is not directly affected by other agents' strategies but rather depends on the aggregate of all agents' strategies and its own strategy. Nevertheless, Nash equilibrium (NE) is still a widely-used concept in aggregative games, representing a stable and desirable solution where no agent can improve its payoff through unilateral deviation. Recently, the NE seeking algorithm design in aggregative games has attracted considerable attention \citep{Par20,Deng22,Huang23,Li25}. However, since aggregative games primarily consider agents in competition, they are inadequate for addressing real-life scenarios where competition and cooperation coexist, which drives the development of multi-cluster aggregative games (MAGs).

In MAGs, agents in the same cluster collaborate to optimize the sum of costs, while each cluster is viewed as a virtual self-interested player to optimize its own cost. Additionally, the cost of each agent is influenced by its own strategy and the aggregate of all agents' strategies in the game, regardless of their cluster affiliation. When each cluster contains only one agent, the MAG reduces to a standard aggregative game \citep{Kos16}. On the other hand, when there is just one cluster, it reduces to an aggregative optimization problem \citep{Li22,Chen23}. It is worth noting that existing studies on the coexistence of cooperation and competition mainly focus on multi-cluster games \citep{Meng23,Ngu24,Lif24,Huang25,Liu25}, a class of games obtained by incorporating cooperative mechanisms into general non-cooperative games. Although these methods can be applied to MAGs, they generally require each agent to estimate the strategies of other agents and then compute the aggregate, which inevitably increases computational and communication burdens. In fact, for MAGs, it is sufficient to obtain only the aggregate, without access to the specific strategies of each agent. Specifically, an algorithm based on best response is proposed for MAGs in \citet{Keb21}, where a cluster coordinator is used to gather information within its own cluster and exchange local aggregate information with the neighboring coordinators. A similar communication architecture is adopted in \citet{Chen24} to develop a gradient-based algorithm. Further, fully distributed continuous-time and discrete-time NE seeking algorithms are developed in \citet{Huang24} and \citet{Zhao26}, respectively, where each agent estimates the aggregate and gradient information required for NE seeking using only information exchanged with its neighbors over intra-cluster and inter-cluster networks.

In contrast to existing gradient-based algorithms for MAGs in the Euclidean setting, this paper adopts the mirror descent method, a generalization of gradient descent, to develop a distributed algorithm for MAGs under a flexible non-Euclidean framework induced by the Bregman divergence. This framework provides greater flexibility in selecting a distance-generating function suited to the geometry of the constraint set, potentially leading to more efficient strategy updates. It has already been widely applied to optimization problems. For example, \citet{Doan19} analyzes centralized and distributed mirror descent algorithms. A distributed zeroth-order mirror descent method is proposed in \citet{Yu22}. To reduce communication costs, an event-triggered strategy is employed to design the distributed mirror descent algorithm in \citet{Xiong23}. More recently, \citet{Wang24} extends mirror descent to NE seeking in stochastic aggregative games. To the best of our knowledge, there is limited literature exploring how to apply the mirror descent method to handle MAGs where cooperation and competition coexist. Moreover, this problem is not simply a combination of the mirror descent algorithms from distributed optimization or game problems. Instead, several challenges arise from coordinating information exchange in the primal space with gradient updates in the dual space and handling the asymmetry of the Bregman divergence in the convergence analysis. Therefore, developing a mirror descent algorithm for MAGs is of great significance from both theoretical and practical perspectives.

In time-varying networks, sparse communication links may slow information transmission, thereby affecting the convergence performance of algorithms. To alleviate this effect, acceleration techniques are considered an effective solution. In this regard, the Nesterov acceleration method has been widely applied to address distributed optimization \citep{Kri15,Liu24} and game problems \citep{Tat21}, and has also been extended to more complex scenarios where cooperation and competition coexist \citep{Wang25}. However, the assumptions on the pseudo-gradient generally need to hold over $\mathbb{R}^{md}$, since the accelerated search points may lie outside the feasible region. To avoid this limitation, the gradient extrapolation provides an alternative acceleration mechanism by performing extrapolation in the dual space rather than in the primal space, as illustrated in Fig.~\ref{fig_1}. Accordingly, gradient-extrapolation-based methods are developed for variational inequality problems \citep{Kot22}, optimization problems \citep{Lan18}, and game problems \citep{Wang24}, demonstrating improved convergence performance. Nevertheless, extending this acceleration mechanism to MAGs remains challenging due to the coexistence of cooperation and competition.

\begin{figure}[!htp]
	\centering
	\includegraphics[width=3in]{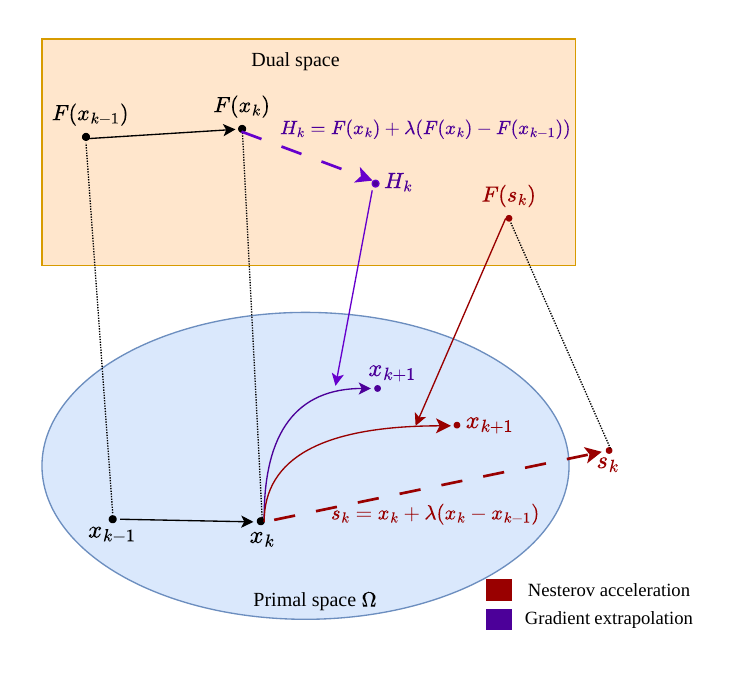}
	\caption{Comparison of acceleration mechanisms between Nesterov acceleration method and gradient extrapolation method.}
	\label{fig_1}
\end{figure}

Inspired by the above discussions, this paper aims to develop a gradient-extrapolation-based distributed mirror descent algorithm for MAGs, capable of addressing complex scenarios where cooperation and competition coexist and achieving efficient NE seeking The main contributions are presented as follows:
\begin{itemize}
	\item A distributed mirror descent algorithm incorporating gradient extrapolation is proposed for efficient NE seeking in MAGs. In contrast to existing algorithms for MAGs \citep{Chen24,Huang24,Zhao26} that rely on the classical Euclidean distance, the proposed algorithm employs the more general Bregman divergence within the mirror descent framework, thereby offering greater flexibility. Moreover, the gradient extrapolation method leverages historical gradient information in the decision-making process to improve the convergence performance.
	\item Compared with existing methods, as summarized in Table~\ref{tab1}, the proposed mirror descent algorithm achieves convergence over time-varying intra-cluster and inter-cluster networks under the restricted strong monotonicity characterized by the Bregman divergence. Furthermore, the convergence rate of $\mathcal{O}(1/k)$ is established through an appropriate design of the step sizes and algorithm parameters.
\end{itemize}

\begin{table*}[t]  
	\centering
	\caption{Recent studies on multi-cluster aggregative games.}
	\label{tab1}
	\begin{threeparttable} 
	\begin{tabular}{cccccc}  
		\hline
	     Literature & Network & \makecell[c]{Aggregate acquisition\\scheme} & Algorithm &  \makecell[c]{Strong monotonicity\\condition}   & Rate\\
	     \hline
		 \citet{Keb21} & Time-varying & Semi-decentralized &  Best-response & $-$  & $-$\\
		 \citet{Chen24} & Fixed & Semi-decentralized & Gradient &   Euclidean-based  & Linear\\
		 \citet{Huang24} & Fixed & Distributed & Gradient & Euclidean-based  & $-$\\
		 \citet{Zhao26} & Fixed & Distributed & Gradient & Euclidean-based & Linear\\
		 This paper  & Time-varying & Distributed & Mirror descent & Bregman-based restricted  & $\mathcal{O}(1/k)$\\
		 \hline
	\end{tabular}
	\begin{tablenotes}
		\item[] \small A dash ($-$) indicates that the corresponding property is not considered in the literature.
	\end{tablenotes}
	\end{threeparttable}
\end{table*}

\emph{Notations}: Denote $\lceil h \rceil:=\{1,2,\ldots,h\}$ for an integer $h\geq1$. All vectors are regarded as column vectors unless otherwise specified. The symbol $col((x_{j})_{j\in \mathcal{M}})$ is the stacked vector in the form of $[(x_{1})^{\top},\cdots,(x_{m})^{\top}]^{\top}$. For a vector $b$, $[b]_{i}$ denotes its $i$th element, and $b^{\top}$ represents its transpose. $\langle a,b\rangle=a^{\top}b$ represents the inner product of vectors $a$ and $b$. The Euclidean norm of vector $a$ is denoted as $\lVert a \rVert=\sqrt{a^{\top}a}$. Also, the dual norm with respect to $\lVert \cdot \rVert$ is given by $\lVert z \rVert_{*}=\max_{\lVert a\rVert=1}\langle z,a\rangle$. Define $\mathbf{1}_{m}$ as the $m$-dimensional vector with all entries equal to $1$, and $\mathbf{I}$ as the identity matrix.

\section{Problem Formulation}
\label{sec:PF}

\subsection{Multi-cluster Aggregative Games}
A MAG with an agent set $\mathcal{M}\triangleq\{1,2,\cdots,m\}$ consists of $h$ clusters. Every cluster $l\in \lceil h \rceil$, viewed as a virtual player, has an agent set $\mathcal{M}^{l}$ containing $m^{l}$ agents, where $m=\sum_{l=1}^{h}m^{l}$ and agents across different clusters are distinct from one another, i.e. $\mathcal{M}^{l}\cap\mathcal{M}^{l'}=\emptyset$ for any $l\neq l'\in \lceil h \rceil$. For notational clarity, agent indices are represented by subscripts, whereas cluster indices are indicated through superscripts.

Every agent $j\in \mathcal{M}^{l}$ in any cluster $l \in \lceil h \rceil$ has its own cost function  which is privately accessible only to itself. This function $f_{j}^{l}(x_{j}^{l},\psi(\mathbf{x}))$ depends on the agent's own strategy $x^{l}_{j}$ chosen from a strategy set $\Omega_{j}^{l}\subset \mathbb{R}^d$, and the aggregate $\psi$ of all agents' strategies $\mathbf{x}\triangleq col(((x_{j}^{l})_{j\in\mathcal{M}^{l}})_{l\in\lceil h \rceil})\in \mathbb{R}^{md}$. Specifically, the aggregate $\psi(\mathbf{x})$ is defined as $\psi(\mathbf{x})\triangleq\sum_{l=1}^{h}\sum_{j=1}^{m^{l}}x_{j}^{l}\in \mathbb{R}^{d}$. In MAGs, all agents within any given cluster $l\in \lceil h \rceil$ collaborate to minimize the cluster's cost function $f^{l}(x^{l},\psi(\mathbf{x}))$ formed by the sum of these agents' cost functions, as expressed by
\begin{equation}
	\mathop{\min}_{x^{l}\in \Omega^{l}}f^{l}(x^{l},\psi(\mathbf{x}))=\sum_{j=1}^{m^{l}}f_{j}^{l}(x_{j}^{l},\psi(\mathbf{x})),
	\label{Eq:game}
\end{equation}
where $x^{l}=col((x_{j}^{l})_{j\in \mathcal{M}^{l}})\in \Omega^{l}$, and $\Omega^{l}\triangleq\prod_{i=1}^{m^{l}}\Omega^{l}_{i}\subset \mathbb{R}^{dm^{l}}$. It is worth noting that problem \eqref{Eq:game} includes two special cases. Specifically, when $h=1$, it reduces to an aggregative optimization problem. When $m^{l}=1$ for all $l\in \lceil h\rceil$ and $h>1$, it reduces to an aggregative game.

The Nash equilibrium (NE) $\mathbf{x}^{*}$ is a stable strategy profile for the MAG satisfying, for any cluster $l\in\lceil h\rceil$,
\begin{equation*}
	f^{l}(x^{l,*},\psi(\mathbf{x}^{*}))\leq f^{l}(x^{l},\psi(x^{l})+\psi(\mathbf{x}^{-l,*})), ~~\forall x^{l}\in\Omega^{l}, 
\end{equation*}
where $\psi(x^{l})\triangleq\sum_{j\in\mathcal{M}^{l}}x_{j}^{l}$, and $\psi(\mathbf{x}^{-l,*})\triangleq\sum_{i\in\lceil h \rceil\backslash\{l\}}$ $\sum_{j\in \mathcal{M}^{i}}x_{j}^{i,*}$ is the aggregate of the strategies of all agents outside cluster $l$. In other words, no cluster can reduce its cost by unilaterally deviating from this profile $\mathbf{x}^{*}$.

The following assumptions on the cost functions are introduced for the subsequent analysis.
\begin{asum}\label{asum1}
For each agent $j$ in any cluster $l\in \lceil h \rceil$, the strategy set $\Omega_{j}^{l}$ is convex, compact and closed. Moreover, the cost function $f_{j}^{l}(x_{j}^{l},\psi(\mathbf{x}))$ is continuously differentiable with respect to $(x_{j}^{l},\psi(\mathbf{x}))\in \Omega_{j}^{l}\times\mathbb{R}^d$ and is convex over $\Omega^{l}$ for any given $\mathbf{x^{-l}}\in \Omega^{-l}\triangleq\prod_{r\neq l}\Omega^{r}$.
\end{asum}	

Based on Assumption \ref{asum1}, it follows that $f^{l}(x^{l},\psi(\mathbf{x}))$ is convex over $\Omega^{l}$ and continuously differentiable in $(x^{l},\psi(\mathbf{x}))\in \Omega^{l}\times\mathbb{R}^d$. For each agent $j$ in cluster $l$,
let $\nabla_{1}f_{j}^{l}(x_{j}^{l},z)$ and $\nabla_{2}f_{j}^{l}(x_{j}^{l},z)$ denote the gradients of
$f_{j}^{l}(x_{j}^{l},z)$ with respect to its first argument
$x_{j}^{l}$ and second argument $z$, respectively. Since, among the local cost functions in cluster $l$, $x_{j}^{l}$ appears explicitly only in the first argument of $f_{j}^{l}$, the direct derivative with respect to $x_{j}^{l}$ involves only $\nabla_{1}f_{j}^{l}$, whereas its indirect effect on the cluster cost through the aggregate is captured by $\nabla_{2}f^{l}=\sum_{i=1}^{m^{l}}\nabla_{2}f_{i}^{l}$. Accordingly, for any cluster $l\in\lceil h\rceil$, the gradient of $f^{l}(x^{l},\psi(\mathbf{x}))$ with respect to $x_{j}^{l}$ is defined as 
\begin{equation*}	
	F_{j}^{l}(\mathbf{x})\triangleq\nabla_{1}f_{j}^{l}(x_{j}^{l},\psi(\mathbf{x}))+\nabla_{2}f^{l}(x^{l},\psi(\mathbf{x})).
\end{equation*}
Denote the pseudo gradient as $F(\mathbf{x})\triangleq col(((F_{j}^{l}(\mathbf{x}))_{j\in\mathcal{M}^{l}})_{l\in\lceil h\rceil})$. Let $\bar{z}_{j}^{l}$ denote the estimate of the aggregate by agent $j$ in cluster $l$ and $\bar{z}^{l}\triangleq col((\bar{z}_{j}^{l})_{j\in \mathcal{M}^{l}})$. Under the estimated aggregates, define
\begin{equation*}	
	G_{j}^{l}(x^{l},\bar{z}^{l})\triangleq\nabla_{1}f_{j}^{l}(x_{j}^{l},\bar{z}_{j}^{l})+\sum_{i=1}^{m^l}\nabla_{2}f^{l}_i(x_i^{l},\bar{z}^{l}_i).
\end{equation*}	
It is easily obtained that $F_{j}^{l}(\mathbf{x})=G_{j}^{l}(x^{l},\mathbf{1}_{m^{l}}\otimes\mathbf{I}_{d}\psi(\mathbf{x}))$. The following Lipschitz continuity conditions are imposed on these mappings.
\begin{asum}\label{asum2}
$\textup{(i)}$ $F(\mathbf{x})$ is $L$-Lipschitz continuous for any $\mathbf{x}_{1}, \mathbf{x}_{2}\in \Omega\triangleq\prod_{l\in\lceil h\rceil}\Omega^{l}$, i.e., $	
	\lVert F(\mathbf{x}_{1})-F(\mathbf{x}_{2})\rVert_{*}\leq L\lVert\mathbf{x}_{1}-\mathbf{x}_{2}\rVert$.\\
$\textup{(ii)}$ For each agent $j$ in cluster $l\in\lceil h\rceil$, $G_{j}^{l}(x^{l},\bar{z}^{l})$ is $\bar{L}$-Lipschitz continuous in $\bar{z}^{l}$ for any fixed $x^{l}\in \Omega^{l}$, i.e., for any $\bar{z}_{1}^{l},\bar{z}_{2}^{l}\in \mathbb{R}^{m^{l}d}$, $\lVert G_{j}^{l}(x^{l},\bar{z}_{1}^{l})\!-\!G_{j}^{l}(x^{l},\bar{z}_{2}^{l})\rVert_{*}\!\leq\! \bar{L}\lVert\bar{z}_{1}^{l}\!-\!\bar{z}_{2}^{l}\rVert$.\\
$\textup{(iii)}$ For each agent $j$ in cluster $l\in\lceil h\rceil$, $\nabla_{2}f_{j}^{l}(x_{j}^{l},\bar{z}_{j}^{l})$ is $\hat{L}$-Lipschitz continuous over $\Omega_{j}^{l}\times\mathbb{R}^{d}$, i.e., $\lVert \nabla_{2}f_{j}^{l}(x_{j1}^{l},\bar{z}_{j1}^{l})-\nabla_{2}f_{j}^{l}(x_{j2}^{l},\bar{z}_{j2}^{l})\rVert_{*}\leq \hat{L}(\lVert x_{j1}^{l}-x_{j2}^{l}\rVert+\lVert\bar{z}_{j1}^{l}-\bar{z}_{j2}^{l}\rVert)$ for any $x_{j1}^{l},x_{j2}^{l}\in\Omega_{j}^{l}$ and $\bar{z}_{j1}^{l},\bar{z}_{j2}^{l}\in\mathbb{R}^{d}$.
\end{asum}	
\begin{rem}
It is worth noting that Assumption \ref{asum1} ensures the existence of NE for MAGs \citep{Scu10}, and under Assumption \ref{asum1}, the NE $\mathbf{x}^{*}$ can be equivalently characterized as the solution to the variational inequality problem $\mathrm{VI}(\Omega,F)$ i.e., $(\mathbf{x}'\!-\!\mathbf{x}^{*})^{\top}F(\mathbf{x}^{*})\!\geq\!0$ for any $\mathbf{x}'\in \Omega$.
\end{rem}

In the designed non-Euclidean algorithms, the Bregman divergence is specified as follows.
\begin{defn}\label{defn1}
Let $\phi(\cdot)$ be a $1$-strongly convex differentiable function with respect to the norm $\lVert\cdot\rVert$, i.e., $\phi(x_{2})\geq \phi(x_{1})+\langle\nabla \phi(x_{1}),x_{2}-x_{1}\rangle+\frac{1}{2}\lVert x_{1}-x_{2}\rVert^2$ for any $x_{1},x_{2}\in\bar{\Omega}$, where $\bar{\Omega}$ is the constraint set. The Bregman divergence $D_{\phi}(x_{1},x_{2})$ associated with $\phi$ is given by
\begin{equation}
	D_{\phi}(x_{1},x_{2})=\phi(x_{2})-\phi(x_{1})-\langle\nabla \phi(x_{1}),x_{2}-x_{1}\rangle.
	\label{eq:breg}
\end{equation}
\end{defn}

From Definition \ref{defn1}, the following relationship between the Bregman divergence and the Euclidean distance is established:
\begin{equation}
	D_{\phi}(x_{1},x_{2})\geq\frac{1}{2}\lVert x_{1}-x_{2}\rVert^2,~~\forall x_{1},x_{2}\in\bar{\Omega}.
	\label{eq:relat}
\end{equation}

\begin{rem}
The function $\phi(\cdot)$  is assumed to be $1$-strongly convex for the sake of presenting the main results more clearly. If the $\vartheta$-strongly convex assumption is adopted, the convergence analysis remains essentially unchanged.
\end{rem}

In existing studies \citep{Chen24,Pu21,Ngu23,Zhao26}, strong monotonicity is commonly formulated in terms of the Euclidean distance. In contrast, this paper formulates the condition using the Bregman divergence, making it better aligned with the non-Euclidean geometry underlying mirror descent. Moreover, the condition is imposed in a restricted form relative to the NE $\mathbf{x}^{*}$. Specifically, for any $\mathbf{x}_{1},\mathbf{x}_{2}\in\Omega$, define $\bar{D}_{\phi}(\mathbf{x}_{1},\mathbf{x}_{2})=\sum_{l\in\lceil h\rceil}\sum_{j\in\mathcal{M}^{l}}D_{\phi}(x^{l}_{j,1},x^{l}_{j,2})$. This assumption is then stated as follows.
\begin{asum}\label{asum3}
The mapping $F(\mathbf{x})$	is restricted strongly monotone with respect to the NE $\mathbf{x}^{*}$ with a constant $\mu>0$, that is, for any $\mathbf{x}\in\Omega$,
$\langle F(\mathbf{x})-F(\mathbf{x}^{*}),\mathbf{x}-\mathbf{x}^{*}\rangle\geq2\mu\bar{D}_{\phi}(\mathbf{x},\mathbf{x}^{*})$.
\end{asum}

\subsection{Communication Networks}

For MAGs, two distinct types of networks are involved in the design of the fully distributed algorithm: local intra-cluster networks for tracking the cluster gradient and the global inter-cluster networks for tracking the aggregate. Specifically, for the intra-cluster network $\mathcal{G}^{l}_{k}=\{\mathcal{M}^{l},\mathcal{E}^{l}_{k}\}$, where $\mathcal{E}^{l}_{k}\subseteq\mathcal{M}^{l}\times\mathcal{M}^{l}$ is the edge set, $(i,j)\in\mathcal{E}^{l}_{k}$ indicates that agent $j$ receives information from agent $i$ within cluster $l$ at time $k$. Correspondingly, the adjacency matrix $C_{k}^{l}\in\mathbb{R}^{m^{l}\times m^{l}}$ satisfies $[C_{k}^{l}]_{ji}>0$ if $(i,j)\in\mathcal{E}^{l}_{k}$ and $[C_{k}^{l}]_{ji}=0$ otherwise. Similarly, for the inter-cluster network $\mathcal{G}^{0}_{k}=\{\mathcal{M}^{0},\mathcal{E}^{0}_{k}\}$ with $\mathcal{M}^{0}=\mathcal{M}$ and $\mathcal{E}^{0}_{k}\subseteq\mathcal{M}\times\mathcal{M}$, the adjacency matrix $C^{0}_{k}\in \mathbb{R}^{m\times m}$ satisfies $[C_{k}^{0}]_{ji}^{lu}>0$ if $(i,j)\in\mathcal{E}^{0}_{k}$ and $[C_{k}^{0}]_{ji}^{lu}=0$ otherwise, where $l$ and $u$ denote the clusters to which agents $j$ and $i$ belong, respectively. 

\begin{asum}\label{asum4}
There exists a positive integer $Q$ such that, for each $i\in\{0\}\cup\lceil h\rceil$ and any $p\geq0$, the union graph $\{\mathcal{M}^{i},\cup_{k=pQ}^{(p+1)Q-1}\mathcal{E}^{i}_{k}\}$ is strongly connected. Each adjacency matrix $C^{i}_{k}$ is doubly stochastic for any $k\geq0$, i.e.,
$C^{0}_{k}\mathbf{1}_{m}=\mathbf{1}_{m}, \mathbf{1}_{m}^{\top}C^{0}_{k}=\mathbf{1}_{m}^{\top},C^{l}_{k}\mathbf{1}_{m^{l}}=\mathbf{1}_{m^{l}}, \mathbf{1}_{m^{l}}^{\top}C^{l}_{k}=\mathbf{1}_{m^{l}}^{\top},l\in\lceil h\rceil$.
Moreover, for any $k\geq0$, each $C^{i}_{k}$ has positive diagonal elements, and there exists a constant $0<\gamma<1$ such that $[C_{k}^{l}]_{js}\geq\gamma$ for all $s\in\mathcal{N}^{l}_{k,j}$ and $l\in\lceil h\rceil$, and $[C_{k}^{0}]_{js}^{lu}\geq\gamma$ for all $s\in\mathcal{N}^{0}_{k,j}$, where $\mathcal{N}^{i}_{k,j}\triangleq\{s\in \mathcal{M}^{i}\mid(s,j)\in\mathcal{E}^{i}_{k}\}$. 
\end{asum}	

Unlike the assumptions in \citet{Chen24,Huang24,Zhao26}, Assumption~\ref{asum4} allows the communication digraphs to be disconnected at some time, thereby covering a broader class of communication scenarios. For any $k\geq t\geq0$, let $\Xi^{i}_{(k,t)}=C^{i}_{k}C^{i}_{k-1}\cdots C^{i}_{t}$. The following property is then stated to facilitate the subsequent analysis.
\begin{lem}\citep{Ned09}\label{lem1}
Under Assumption \ref{asum4}, there exist $\beta^{0}\triangleq(1-\gamma/(4m^{2}))^{-2}>0$, $0<\eta^{0}\triangleq(1-\gamma/(4m^{2}))^{1/Q}<1$, and for each $l\in\lceil h\rceil$, $\beta^{l}\triangleq(1-\gamma/(4(m^l)^{2}))^{-2}>0$, $0<\eta^{l}\triangleq(1-\gamma/(4(m^l)^{2}))^{1/Q}<1$, such that for any $k\geq t\geq0$,
\begin{equation*}
\begin{aligned}
	&\mid [\Xi^{0}_{(k,t)}]_{js}-1/m\mid\leq \beta\eta^{k-t}, ~~\forall j,s\in\mathcal{M}, \nonumber\\
	&\mid [\Xi^{l}_{(k,t)}]_{j's'}-1/m^{l}\mid\leq \beta\eta^{k-t}, ~~\forall j',s'\in\mathcal{M}^{l},
\end{aligned}
\end{equation*}
where $\beta=\max_{i\in\{0,1,\cdots,h\}}\beta^{i}$ and $\eta=\max_{i\in\{0,1,\cdots,h\}}\eta^{i}$.
\end{lem}

\section{Algorithm Design}
\label{sec:AD}

In this section, a distributed mirror descent algorithm with gradient extrapolation is proposed to seek the NE of MAGs. Compared with the settings in \citet{Lan18} and \citet{Kot22}, an important challenge in the considered setting is that each agent can only access local information exchanged through intra-cluster and inter-cluster networks. Therefore, each agent $j$ in cluster $l\in\lceil h\rceil$ maintains a local variable $x_{j,k}^{l}\in \Omega_{j}^{l}$ as an estimate of its equilibrium strategy and an estimate $v_{j,k}^l$ to track the aggregate $\psi(\mathbf{x}_{k})$. Moreover, an auxiliary variable $y_{j,k}^{l}$ is introduced to estimate the derivative with respect to the second argument, since this term involves information from all agents within cluster $l$, whereas the derivative of $f_j^l$ with respect to the first argument can be computed locally. Based on these variables, the update procedure of the proposed algorithm is given in Algorithm~\ref{alg:alg1}.

\begin{algorithm}[!htb]
	\caption{Gradient-extrapolation-based distributed mirror descent algorithm}\label{alg:alg1}
		Define the step size $\alpha_{k}>0$ and parameter $\lambda_{k}>0$ at time $k$.\\
		{\textbf{Initialization:}} For any agent $j$ of any cluster $l\in\lceil h \rceil$, initialize with $x_{j,0}^{l}=x_{j,-1}^{l}\in \Omega_{j}^{l}$, and let $v_{j,0}^{l}=v_{j,-1}^{l}=x_{j,0}^{l}$, $y_{j,-1}^{l}=y_{j,0}^{l}=\nabla_{2}f_{j}^{l}(x_{j,0}^{l},mv_{j,0}^{l})$.\\
		{\textbf{Iteration:}} For $k\geq0$, each agent $j$ in cluster $l\in\lceil h\rceil$ executes the following update:
		\begin{subequations}
			\begin{align}
			&x_{j,k+1}^{l}\!=\!\mathop{\operatorname{argmin}}\limits_{\bar{x}_{j}^{l}\in\Omega_{j}^{l}}\Big\{
			\alpha_{k}\langle H_{j,k}^{l} ,\bar{x}_{j}^{l}\rangle+D_{\phi}(x_{j,k}^{l},\bar{x}_{j}^{l})\Big\},\label{Eq:x}\\ 
			&v_{j,k+1}^l=\sum_{i\in \mathcal{N}^{0}_{k,j}} [C_{k}^{0}]_{ji}^{lu} v_{i,k}^{u}+x_{j,k+1}^{l}-x_{j,k}^{l},\label{Eq:v}\\ 
			&y_{j,k+1}^{l}=\sum_{i=1}^{m^{l}}[C_{k}^{l}]_{ji}y_{i,k}^{l}+\nabla_{2}f_{j}^{l}(x_{j,k+1}^{l},mv_{j,k+1}^{l})\nonumber\\
			&~~~~~~~~~~~~~~~~~~-\nabla_{2}f_{j}^{l}(x_{j,k}^{l},mv_{j,k}^{l}),\label{Eq:y}
\end{align}
\end{subequations}
where $H_{j,k}^{l}=(1+\lambda_k)(\nabla_{1}f_{j}^{l}(x_{j,k}^{l},mv_{j,k}^{l})+m^{l}y_{j,k}^{l})-\lambda_k$ $ (\nabla_{1}f_{j}^{l}(x_{j,k-1}^{l},mv_{j,k-1}^{l})+m^{l}y_{j,k-1}^{l})$.
\end{algorithm}

At each iteration, each agent $j$ in cluster $l$ updates its strategy $x_{j,k+1}^{l}$ by employing the mirror descent scheme with gradient extrapolation. Different from the Nesterov acceleration method, where extrapolation is performed in the primal space, the gradient extrapolation used here is carried out in the dual space. Specifically, the gradient extrapolation term
$\lambda_{k}([\nabla_{1}f_{j}^{l}(x_{j,k}^{l},mv_{j,k}^{l})+m^{l}y_{j,k}^{l}]-[\nabla_{1}f_{j}^{l}(x_{j,k-1}^{l},mv_{j,k-1}^{l})+m^{l}y_{j,k-1}^{l}])$ 
is incorporated into the strategy update to improve the convergence performance. Since the aggregate $\psi(\mathbf{x}_{k})$ depends on the strategies of all agents in MAGs, the aggregate estimates are exchanged among neighboring agents through the inter-cluster network $\mathcal{G}^{0}_{k}$, regardless of their cluster affiliations. Based on the received information and the strategy increment $x_{j,k+1}^{l}-x_{j,k}^{l}$, each agent $j$ in cluster $l$ updates its estimate $v_{j,k+1}^{l}$ to track the aggregate $\psi(\mathbf{x}_{k+1})$. In addition, the estimate $y_{j,k}^{l}$ of agent $j\in \mathcal{M}^{l}$ is to track $\frac{1}{m^{l}}\sum_{i=1}^{m^{l}}\nabla_{2}f_{i}^{l}(x_{i,k}^{l},mv_{i,k}^{l})$ which is only related to the agents within the cluster $l$. Thus, $y_{j,k+1}^{l}$ is updated using information received through the intra-cluster network and the local derivative increment $\nabla_{2}f_{j}^{l}(x_{j,k+1}^{l},mv_{j,k+1}^{l})-\nabla_{2}f_{j}^{l}(x_{j,k}^{l},mv_{j,k}^{l})$.

\begin{rem}
In the mirror descent scheme, an appropriate distance function can be selected to match the characteristics of the constraint set, which may yield a direct and computationally efficient update while reducing the cost associated with Euclidean projection. For example, for the simplex constraint set, i.e., $\Omega_{j}^{l}\triangleq\{x_{j}^{l}\geq0: \mathbf{1}_{d}^{\top}x_{j}^{l}=1\}$, the distance function can be selected as  $\phi(x_{j}^{l})=\sum_{i=1}^{d}[x_{j}^{l}]_{i}\ln[x_{j}^{l}]_{i}$. The corresponding Bregman divergence is $D_{\phi}(x_{j1}^{l},x_{j2}^{l})=\sum_{i=1}^{d}[x_{j2}^{l}]_{i}\ln\frac{[x_{j2}^{l}]_{i}}{[x_{j1}^{l}]_{i}}$, which is known as the Kullback-Leibler divergence. In this case, the update \eqref{Eq:x} takes the following form:
\begin{equation*}
[x_{j,k+1}^{l}]_i = \frac{
	[x_{j,k}^{l}]_i \exp\big( -\alpha_k [H_{j,k}^{l}]_{i} \big)
}{\sum\limits_{r=1}^{d} [ x_{j,k}^{l}]_r \exp\Big( -\alpha_k [H_{j,k}^{l}]_r \Big)
}, ~~\forall i\in \lceil d \rceil.	
\end{equation*}
Under the simplex constraint, this update has a computational complexity of $\mathcal{O}(d)$ per iteration, whereas the distributed projected gradient method in \citet{Zhao26} requires $\mathcal{O}(d\log d)$, thereby reducing the per-iteration computational burden. Moreover, the mirror descent scheme includes the Euclidean projected gradient method with extrapolation as a special case. Specifically, when $\phi(x_{j}^{l})=\frac{1}{2}\lVert x_{j}^{l}\rVert^{2}$, the corresponding Bregman divergence becomes $D_{\phi}(x_{j1}^{l},x_{j2}^{l})=\frac{1}{2}\lVert x_{j1}^{l}-x_{j2}^{l}\rVert^{2}$. In this case, the update \eqref{Eq:x} is equivalent to the projected update with extrapolation $x_{j,k+1}^{l}=P_{\Omega_{j}^{l}}[x_{j,k}^{l}-\alpha_k H_{j,k}^{l}]$, where $P_{\Omega_{j}^{l}}$ denotes the Euclidean projection onto the set $\Omega_{j}^{l}$.
\end{rem}

\section{Convergence Analysis}
\label{sec:CA}
In this section, the convergence analysis of the designed algorithm is provided in detail.
\subsection{Preliminary Results}
This subsection shows several auxiliary results that are crucial for establishing the subsequent convergence results. Firstly, for each agent $j$ within cluster $l\in\lceil h \rceil$, bounds are established for the aggregate tracking error $Ev^{l}_{j,k+1}=\lVert  v^{l}_{j,k+1}-\frac{\psi(\mathbf{x}_{k+1})}{m}\rVert$, as well as the gradient tracking error $Ey^l_{j,k+1}=\lVert y_{j,k+1}^{l}-\frac{1}{m^{l}}\sum_{i=1}^{m^{l}}\nabla_{2}f_{i}^{l}(x_{i,k+1}^{l},mv_{i,k+1}^{l}) \rVert_{*}$.

\begin{lem} \label{lem2}
Under Assumptions \ref{asum1}, \ref{asum2} and \ref{asum4}, suppose that, for any agent $j$ in cluster $l\in\lceil h \rceil$, the initial conditions $v_{j,0}^{l}=x_{j,0}^{l}$ and $y_{j,0}^{l}=\nabla_{2}f_{j}^{l}(x_{j,0}^{l},mv_{j,0}^{l})$ hold. Then, the aggregate and gradient tracking errors satisfy, respectively,
\begin{align}
		&Ev^{l}_{j,k+1} \leq\beta\eta^{k}m\hat{N}+m\tilde{W}\beta\sum_{s=1}^{k}\eta^{k-s}\alpha_{s-1}(1+2\lambda_{s-1})+\nonumber\\
		&~~2\alpha_{k}(1+2\lambda_{k})\tilde{W}, \label{eq:verr}\\
		&Ey^{l}_{j,k+1}\leq\beta\eta^{k}m^{l}N_{f}\!+\!\beta\eta^{k-1}\hat{L}m^{l}
		\Big\{(m\!+\!1)\alpha_{0}(1\!+\!2\lambda_{0})\tilde{W}\nonumber\\
		&~~\!+\!2m\hat{N}\Big\}\!+\!\beta\!\sum_{s=2}^{k}\eta^{k\!-\!s}\hat{L}m^{l}\Big\{(m\!+\!1)\alpha_{s\!-\!1}(1\!+\!2\lambda_{s\!-\!1})\tilde{W}\!+\!4m\nonumber\\
		&~~\alpha_{s\!-\!2}(1\!+\!2\lambda_{s\!-\!2})\tilde{W}\!+\!2m^{2}\big(\beta\eta^{s\!-\!2}\hat{N}\!+\!\tilde{W}\beta\!\sum_{r=1}^{s\!-\!2}\eta^{s\!-\!2\!-\!r}\alpha_{r\!-\!1}\nonumber\\
		&~~(1\!+\!2\lambda_{r\!-\!1})\big)\Big\}\!+\!2\hat{L}\Big\{(m\!+\!1)\alpha_{k}(1\!+\!2\lambda_{k})\tilde{W}\!+\!4m\alpha_{k\!-\!1}(1\!+\!2\lambda_{k\!-\!1})\nonumber\\
		&~~\tilde{W}\!+\!2m^{2}\big(\beta\eta^{k\!-\!1}\hat{N}\!+\!\tilde{W}\beta\!\sum_{s=1}^{k\!-\!1}\eta^{k\!-\!1\!-\!s}\alpha_{s\!-\!1}(1\!+\!2\lambda_{s\!-\!1})\big)\Big\}\label{eq:yerr},
	\end{align}
where $\tilde{W}\triangleq \bar{m}\tilde{W}_{y}+\bar{L}m\bar{m}\tilde{W}_{v}+\max\limits_{l\in\lceil h\rceil, j\in\mathcal{M}^{l}}\max\limits_{\mathbf{x}\in\Omega}\lVert F_{j}^{l}(\mathbf{x})\rVert_{*}$ with $\bar{m}\triangleq\max_{l\in\lceil h \rceil}m^l$, while $\tilde{W}_{v}$ and $\tilde{W}_{y}$ represent the auxiliary uniform bounds on the aggregate and the gradient tracking errors, respectively, as detailed in equations \eqref{eq:vx2} and \eqref{eq:yf2}.
\end{lem}

\noindent\textbf{Proof.} First, one establishes the auxiliary
	uniform bounds $\tilde W_v$ and $\tilde W_y$ for the aggregate
	and gradient tracking errors. Under Assumption \ref{asum4} and the given initial condition, it follows by induction that
\begin{equation}
	\sum\limits_{l=1}^{h}\sum_{j=1}^{m^{l}}v_{j,k}^{l}=\sum\limits_{l=1}^{h}\sum_{j=1}^{m^{l}}x_{j,k}^{l}=\psi(\mathbf{x}_{k}), ~~\forall k\geq0.
	\label{Eq:vagg}
\end{equation}	
Based on the update form \eqref{Eq:v} of $v_{j,k+1}^{l}$, recursively expanding
$Ev^{l}_{j,k+1}$ yields
\begin{equation}
	\begin{aligned}
&Ev^{l}_{j,k+1} \leq\sum\limits_{u=1}^{h}\sum_{i=1}^{m^{u}}\left|\frac{1}{m}-[\Xi^{0}_{(k,0)}]_{ji}^{lu}\right|\lVert v_{i,0}^{u}\rVert+\sum_{s=1}^{k}\sum\limits_{u=1}^{h}\sum_{i=1}^{m^{u}}\\
&~~\left|\frac{1}{m}-[\Xi^{0}_{(k,s)}]_{ji}^{lu}\right|\lVert x_{i,s}^{u}-x_{i,s-1}^{u}\rVert+\frac{1}{m}\sum\limits_{u=1}^{h}\sum_{i=1}^{m^{u}}\lVert x_{i,k+1}^{u}\\
&~~-x_{i,k}^{u}\rVert+\lVert x_{j,k+1}^{l}-x_{j,k}^{l}\rVert\\
&\leq\beta\eta^{k}\sum\limits_{u=1}^{h}\sum_{i=1}^{m^{u}}\lVert x_{i,0}^{u}\rVert+\beta\sum_{s=1}^{k}\eta^{k-s}\sum\limits_{u=1}^{h}\sum_{i=1}^{m^{u}}\lVert x_{i,s}^{u}\!-\!x_{i,s-1}^{u}\rVert\\
&~~+\frac{1}{m}\sum\limits_{u=1}^{h}\sum_{i=1}^{m^{u}}\lVert x_{i,k+1}^{u}\!-\!x_{i,k}^{u}\rVert\!+\!\lVert x_{j,k+1}^{l}\!-\!x_{j,k}^{l}\rVert,
\end{aligned}
\label{eq:vx}
\end{equation}	
where the second inequality follows from Lemma~\ref{lem1}. Under Assumption \ref{asum1}, define $\hat{N}\triangleq\max\limits_{l\in\lceil h\rceil, j\in\mathcal{M}^{l}}\max\limits_{x_{j}^{l}\in\Omega_{j}^{l}}\left\|x_{j}^{l}\right\|$, then one has
\begin{align}
Ev^{l}_{j,k+1}& \leq m\hat{N}\beta\eta^{k}\!+\!2m\hat{N}\beta\sum_{s=1}^{k}\eta^{k-s}\!+\!4\hat{N}\nonumber\\
&\leq \underbrace{m\hat{N}\beta\!+\!2m\hat{N}\beta/(1-\eta)\!+\!4\hat{N}}_{\triangleq\tilde{W}_{v}}\label{eq:vx2}.
\end{align}	
Thus, $\tilde W_v$ provides an auxiliary uniform bound on the
	aggregate tracking error. Similarly, based on Assumption \ref{asum4} and the initial condition $y_{j,0}^{l}=\nabla_{2}f_{j}^{l}(x_{j,0}^{l},mv_{j,0}^{l})$, there exists for any cluster $l\in\lceil h \rceil$,
\begin{equation}
	\sum_{j=1}^{m^{l}}y_{j,k}^{l}=\sum_{j=1}^{m^{l}}\nabla_{2}f_{j}^{l}(x_{j,k}^{l},mv_{j,k}^{l}), ~~\forall k\geq0.
	\label{Eq:ygra}
\end{equation}	
Then by considering the update form \eqref{Eq:y} of $y_{j,k+1}^{l}$, one can obtain
\begin{equation}
	\begin{aligned}
	&Ey^l_{j,k+1}\leq\!\sum_{i=1}^{m^{l}}\left|\frac{1}{m^{l}}\!-\![\Xi^{l}_{(k,0)}]_{ji}\right|\lVert y_{i,0}^{l}\rVert_{*}\!+\!\sum_{s=1}^{k}\sum_{i=1}^{m^{l}}\left|\frac{1}{m^{l}}\right.\\
	&~~\left.-[\Xi^{l}_{(k,s)}]_{ji}\right|\lVert \nabla_{2}f_{i}^{l}(x_{i,s}^{l},mv_{i,s}^{l})\!-\!\nabla_{2}f_{i}^{l}(x_{i,s\!-\!1}^{l},mv_{i,s\!-\!1}^{l})\rVert_{*}\\
	&~~\!+\!\frac{1}{m^{l}}\sum_{i=1}^{m^{l}}\lVert\nabla_{2}f_{i}^{l}(x_{i,k\!+\!1}^{l},mv_{i,k\!+\!1}^{l})-\nabla_{2}f_{i}^{l}(x_{i,k}^{l},mv_{i,k}^{l})\rVert_{*}\\
&~~\!+\!\lVert \nabla_{2}f_{j}^{l}(x_{j,k\!+\!1}^{l},mv_{j,k\!+\!1}^{l})\!-\!\nabla_{2}f_{j}^{l}(x_{j,k}^{l},mv_{j,k}^{l})\rVert_{*},\\
&\leq\beta\eta^{k}\sum_{i=1}^{m^{l}}\lVert \nabla_{2}f_{i}^{l}(x_{i,0}^{l},mv_{i,0}^{l})\rVert_{*}\!+\!\beta\sum_{s=1}^{k}\eta^{k-s}\hat{L}\sum_{i=1}^{m^{l}}\big(\lVert x_{i,s}^{l}\!-\\
&~~x_{i,s\!-\!1}^{l}\rVert\!+\!m\lVert v_{i,s}^{l}\!-\!v_{i,s\!-\!1}^{l}\rVert \big)\!+\hat{L}\!\big(\lVert x_{j,k\!+\!1}^{l}\!-\!x_{j,k}^{l}\rVert\!+\!m\lVert v_{j,k\!+\!1}^{l}\\
&~~\!-\!v_{j,k}^{l}\rVert\big)\!+\!\frac{\hat{L}}{m^{l}}\sum_{i=1}^{m^{l}}\big(\lVert x_{i,k\!+\!1}^{l}\!-\!x_{i,k}^{l}\rVert\!+\!m\lVert v_{i,k\!+\!1}^{l}\!-\!v_{i,k}^{l}\rVert\big).
\end{aligned}
\label{eq:yf}	
\end{equation}
Notably, $\|\psi(\mathbf{x}_{k})\|=\|\sum_{j=1}^{m}x_{j,k}\|\leq m\hat{N}$. Together with \eqref{eq:vx2}, this implies that $\|mv_{j,k}^{l}\|\leq m\tilde{W}_{v}+m\hat{N}$ for all $k\geq0$. Moreover, the initialization $v_{j,0}^{l}=x_{j,0}^{l}$ implies that $f_{j}^{l}(x_{j,0}^{l},mv_{j,0}^{l})=f_{j}^{l}(x_{j,0}^{l},mx_{j,0}^{l})$ with $x_{j,0}^{l}\in\Omega_{j}^{l}$. It then follows from Assumption~\ref{asum1} that $\nabla_{2}f_{j}^{l}(x_{j,0}^{l},mv_{j,0}^{l})$ is bounded for any $j$ in cluster $l\in\lceil h \rceil$. Accordingly, define $N_{f}\triangleq\max\limits_{l\in\lceil h \rceil}\max\limits_{j\in\mathcal{M}^{l}}\max\limits_{x_{j,0}^{l}\in\Omega_{j}^{l}}\lVert \nabla_{2}f_{j}^{l}(x_{j,0}^{l},mv_{j,0}^{l})\rVert_{*}$. Further, one derives that
\begin{equation}
	\begin{aligned}
		&Ey^l_{j,k+1}\leq\beta\eta^{k}m^{l}N_{f}\!+\!2\beta\hat{L}m^{l}\sum_{s=1}^{k}\eta^{k-s}\big(\hat{N}+m\tilde{W}_{v}+m\hat{N}\big)\\
		&~~\!+4\hat{L}\!\big(\hat{N}+m\tilde{W}_{v}+m\hat{N}\big)\\
		&\leq \underbrace{\beta \bar{m}N_{f}\!+\!\frac{2\beta\hat{L}\bar{m}\!\big(\hat{N}\!+\!m\tilde{W}_{v}\!+\!m\hat{N}\!\big)}{(1-\eta)}\!+\!4\hat{L}\!\big(\hat{N}\!+\!m\tilde{W}_{v}\!+\!m\hat{N}\!\big)}_{\triangleq\tilde{W}_{y}}.
	\end{aligned}
	\label{eq:yf2}	
\end{equation}
where $\bar{m}\triangleq\max_{l\in\lceil h \rceil}m^l$. Hence, $\tilde{W}_y$ provides an auxiliary uniform bound on the gradient tracking error.

Next, based on the auxiliary uniform bounds
	$\tilde{W}_{v}$ and $\tilde{W}_{y}$, one derives an upper bound
	on the strategy increment $\|x_{j,k+1}^{l}-x_{j,k}^{l}\|$, which is
	subsequently used to establish the iteration-dependent bound for
	the aggregate tracking error. Specifically, the optimality condition of
\eqref{Eq:x}, together with the definition of the Bregman divergence
in \eqref{eq:breg}, gives
\begin{equation*}
	\begin{aligned}
		&(\bar{x}_{j}^{l}-x_{j,k+1}^{l})^{\top}[\alpha_{k}\big((1+\lambda_{k})(\nabla_{1}f_{j}^{l}(x_{j,k}^{l},mv_{j,k}^{l})+m^{l}y_{j,k}^{l})-\\
		&\lambda_{k}(\nabla_{1}f_{j}^{l}(x_{j,k-1}^{l},mv_{j,k-1}^{l})+m^{l}y_{j,k-1}^{l})\big)+\nabla \phi(x_{j,k+1}^{l})\\
		&-\nabla \phi(x_{j,k}^{l})]\geq0, ~~~~~~\forall \bar{x}_{j}^{l}\in\Omega_{j}^{l}.
	\end{aligned}
\end{equation*}
Setting $\bar{x}_{j}^{l}=x_{j,k}^{l}$ and using the 1-strong convexity of $\phi$ yield
\begin{equation}
	\begin{aligned}
		&\lVert x_{j,k}^{l}-x_{j,k+1}^{l} \rVert\leq\alpha_{k}\big((1+\lambda_{k})\lVert\nabla_{1}f_{j}^{l}(x_{j,k}^{l},mv_{j,k}^{l})+\\
		&m^{l}y_{j,k}^{l}\rVert_{*}+\lambda_{k}\lVert\nabla_{1}f_{j}^{l}(x_{j,k-1}^{l},mv_{j,k-1}^{l})+m^{l}y_{j,k-1}^{l}\rVert_{*}\big).
	\end{aligned}
	\label{eq:xdiff}	
\end{equation}
For the term $\lVert\nabla_{1}f_{j}^{l}(x_{j,k}^{l},mv_{j,k}^{l})+m^{l}y_{j,k}^{l}\rVert_{*}$, by adding and subtracting $F_{j}^{l}(\mathbf{x}_{k})$ and $G_{j}^{l}(x^{l}_{k},mv_{k}^{l})$ with $v_{k}^{l}=col((v_{j,k}^{l})_{j\in \mathcal{M}^{l}})$, it can be reformulated as
\begin{equation*}
	\begin{aligned}
		&\lVert\nabla_{1}f_{j}^{l}(x_{j,k}^{l},mv_{j,k}^{l})\!+\!m^{l}y_{j,k}^{l}\rVert_{*}\!=\!\lVert\nabla_{1}f_{j}^{l}(x_{j,k}^{l},mv_{j,k}^{l})\!+\!m^{l}y_{j,k}^{l}\\	
		&~~~~-G_{j}^{l}(x^{l}_{k},mv_{k}^{l})+G_{j}^{l}(x^{l}_{k},mv_{k}^{l})-F_{j}^{l}(\mathbf{x}_{k})+F_{j}^{l}(\mathbf{x}_{k})\rVert_{*}\\
		&=\lVert m^{l}y_{j,k}^{l}-\sum_{i=1}^{m^{l}}\nabla_{2}f_{i}^{l}(x_{i,k}^{l},mv_{i,k}^{l})+G_{j}^{l}(x^{l}_{k},mv_{k}^{l})-\\
		&~~~~G_{j}^{l}(x_k^{l},\mathbf{1}_{m^{l}}\otimes\mathbf{I}_{d}\psi(\mathbf{x}_{k}))+F_{j}^{l}(\mathbf{x}_{k})\rVert_{*}\\
		&\leq m^{l}Ey^l_{j,k}+\bar{L}\sum_{i=1}^{m^{l}}mEv^{l}_{i,k+1}+\lVert F_{j}^{l}(\mathbf{x}_{k})\rVert_{*}\\
		&\leq\underbrace{ \bar{m}\tilde{W}_{y}+\bar{L}m\bar{m}\tilde{W}_{v}+\max\limits_{l\in\lceil h\rceil, j\in\mathcal{M}^{l}}\max\limits_{\mathbf{x}\in\Omega}\lVert F_{j}^{l}(\mathbf{x})\rVert_{*}}_{\triangleq\tilde{W}}.
	\end{aligned}
\end{equation*}
The first inequality follows from the triangle inequality and Assumption~\ref{asum2}. Substituting this result back into \eqref{eq:xdiff} results in
\begin{equation}
		\lVert x_{j,k}^{l}-x_{j,k+1}^{l} \rVert\leq\alpha_{k}(1+2\lambda_{k})\tilde{W}.
	\label{eq:xdiff2}	
\end{equation}
Substituting this bound into the recursive estimate \eqref{eq:vx} for the aggregate tracking error yields \eqref{eq:verr}, as desired.

Finally, based on the preceding results, the explicit
	iteration-dependent bound for the gradient tracking error is derived. Since the recursion
	\eqref{eq:yf} involves the variation
	$\|v_{j,k+1}^{l}-v_{j,k}^{l}\|$, an upper bound on this term needs to be
	established. By introducing $\bar{v}_{k}\triangleq\frac{1}{m}\sum\limits_{l=1}^{h}\sum\limits_{j=1}^{m^{l}}v_{j,k}^{l}=\frac{1}{m}\psi(\mathbf{x}_{k})$ into this term, the following relation is derived for $k\geq1$:
\begin{equation}
	\begin{aligned}
		&m\lVert v_{j,k+1}^{l}\!-\!v_{j,k}^{l}\rVert=m\lVert v_{j,k+1}^{l}\!-\bar{v}_{k}+\bar{v}_{k}\!-v_{j,k}^{l}\rVert\\
		&=m\lVert \sum\limits_{u=1}^{h}\sum_{i=1}^{m^{u}} [C_{k}^{0}]_{ji}^{lu} v_{i,k}^{u}
		+x_{j,k+1}^{l}-x_{j,k}^{l}\!-\bar{v}_{k}+\bar{v}_{k}\!-\!v_{j,k}^{l}\rVert\\
		&\leq m\big(\sum\limits_{u=1}^{h}\sum_{i=1}^{m^{u}}[C_{k}^{0}]_{ji}^{lu}\lVert   v_{i,k}^{u}\!-\!\bar{v}_{k}\rVert\!+\!\lVert v_{j,k}^{l}\!-\!\bar{v}_{k}\rVert\!+\!\lVert x_{j,k\!+\!1}^{l}\!-\!x_{j,k}^{l}\rVert\big)\\
		&\leq m[2\big(\beta\eta^{k-1}m\hat{N}+m\tilde{W}\beta\sum_{s=1}^{k-1}\eta^{k-1-s}\alpha_{s-1}(1+2\lambda_{s-1})\big)\!\\
		&~~+\!4\alpha_{k-1}(1+2\lambda_{k-1})\tilde{W}+\alpha_{k}(1+2\lambda_{k})\tilde{W}],
	\end{aligned}
\label{eq:vdiff}
\end{equation}
where the first inequality follows from Assumption \ref{asum4}, and the last one is obtained by invoking \eqref{eq:verr} and \eqref{eq:xdiff2}. For $k=0$, one has $m\| v_{j,1}^{l}\!-\!v_{j,0}^{l}\|\leq2m\hat{N}+m\alpha_0(1+2\lambda_0)\tilde{W}$. Therefore, the $s=1$ term in \eqref{eq:yf} is treated separately, while \eqref{eq:vdiff} is used to bound the subsequent aggregate estimate variations. Together with \eqref{eq:xdiff2}, substituting these bounds into \eqref{eq:yf} yields the desired iteration-dependent gradient-tracking error bound \eqref{eq:yerr}. $\blacksquare$

Next, a well-known technical result concerning the optimality condition for \eqref{Eq:x} is derived through the application of the three-point identity from the Bregman divergence.	

\begin{lem}\citep{Lan18} \label{lem3}
Let $D_{\phi}$ be defined in \eqref{eq:breg}, the following inequality holds,
\begin{equation}
	\begin{aligned}
		&\alpha_{k}\langle H_{j,k}^{l}, x_{j,k+1}^{l}-\bar{x}_{j}^{l} \rangle+D_{\phi}(x_{j,k}^{l},x_{j,k+1}^{l})\\
		&\leq D_{\phi}(x_{j,k}^{l},\bar{x}_{j}^{l})-D_{\phi}(x_{j,k+1}^{l},\bar{x}_{j}^{l}) ~~~~~~\forall \bar{x}_{j}^{l}\in\Omega_{j}^{l}.
	\end{aligned}
	\label{eq:opt1}	
\end{equation}
\end{lem}
Furthermore, let
\begin{align}
	&\Delta Y_{j,k}^{l}\triangleq(\nabla_{1}f_{j}^{l}(x_{j,k}^{l},mv_{j,k}^{l})+m^{l}y_{j,k}^{l})-\nonumber\\
	&~~~~~~~~~(\nabla_{1}f_{j}^{l}(x_{j,k-1}^{l},mv_{j,k-1}^{l})+m^{l}y_{j,k-1}^{l}),\label{eq:Y}\\
	&\Delta D_{j,k}^{l}(\bar{x}_{j}^{l})\triangleq D_{\phi}(x_{j,k}^{l},\bar{x}_{j}^{l})-D_{\phi}(x_{j,k+1}^{l},\bar{x}_{j}^{l}),\label{eq:Dj}\\
	&\Delta D_{k}(\bar{\mathbf{x}})\!=\!\sum_{l\in\lceil h\rceil}\sum_{j\in\mathcal{M}^{l}}\Delta D_{j,k}^{l}(\bar{x}_{j}^{l})\!=\!\bar{D}_{\phi}(\mathbf{x}_{k},\bar{\mathbf{x}})\!-\!\bar{D}_{\phi}(\mathbf{x}_{k\!+\!1},\bar{\mathbf{x}}), \label{eq:dto}
\end{align}	
where $\bar{\mathbf{x}}\triangleq col(((\bar{x}_{j}^{l})_{j\in \mathcal{M}^{l}})_{l\in\lceil h\rceil})$. Then, a lemma is presented to facilitate the subsequent analysis of the convergence properties of Algorithm \ref{alg:alg1}.

\begin{lem}\label{lem4}
Under Assumptions \ref{asum1} and \ref{asum2}, let $\mathbf{x}_{k}$ be generated by Algorithm \ref{alg:alg1}, and suppose that there exists a positive sequence $\{\epsilon_{k}\}_{k\geq0}$ satisfying
		\begin{align}
			&\epsilon_{k+1}\alpha_{k+1}\lambda_{k+1}=\epsilon_{k}\alpha_{k}\label{eq:ep1}\\
			&\epsilon_{k-1}\geq16\alpha_{k}^{2}\lambda_{k}^{2}L^{2}\epsilon_{k}\label{eq:ep2}.
		\end{align}
Then, for any $\bar{\mathbf{x}}\in\Omega$, one has 
\begin{equation}
	\begin{aligned}
		&\sum_{k=0}^{t}\epsilon_{k}\big(\alpha_{k}\sum\limits_{l=1}^{h}\sum_{j=1}^{m^{l}}\langle \nabla_{1}f_{j}^{l}(x_{j,k\!+\!1}^{l},mv_{j,k+1}^{l})\!+\!m^{l}y_{j,k+1}^{l},x_{j,k\!+\!1}^{l}\!-\!\\
		&\bar{x}_{j}^{l}\rangle\!+\!\bar{D}_{\phi}(\mathbf{x}_{k+1},\bar{\mathbf{x}})\!+\!\frac{1}{2}\bar{D}_{\phi}(\mathbf{x}_{k},\mathbf{x}_{k+1})\big)\!-\!2\epsilon_{t}\alpha_{t}^{2}L^{2}\lVert\mathbf{x}_{t+1}\!-\!\bar{\mathbf{x}}\rVert^{2}\\
		&\leq\sum_{k=0}^{t}\epsilon_{k}\bar{D}_{\phi}(\mathbf{x}_{k},\bar{\mathbf{x}})+\sum_{k=0}^{t}\epsilon_{k}\alpha_{k}\lambda_{k}\delta_{k}+\epsilon_{t}\alpha_{t}\delta_{t\!+\!1},
	\end{aligned}
\label{eq:lem4}	
\end{equation}
where
\begin{equation}
\begin{aligned}
&\delta_{k}\!\triangleq\!2\hat{N}\sum\limits_{l=1}^{h}\!\sum_{j=1}^{m^{l}}\!\big[m^l(Ey^l_{j,k}\!+\!Ey^l_{j,k\!-\!1})\!+\!\bar{L}m\sum_{i=1}^{m^{l}}(Ev^{l}_{i,k}\!+\!Ev^{l}_{i,k\!-\!1})\big].		
\end{aligned}
\label{eq:delta}
\end{equation}		
\end{lem}

\noindent\textbf{Proof.} Based on \eqref{eq:Y} and \eqref{eq:Dj}, the addition and subtraction of $\alpha_{k}\langle \nabla_{1}f_{j}^{l}(x_{j,k+1}^{l},mv_{j,k+1}^{l})+m^{l}y_{j,k+1}^{l},x_{j,k+1}^{l}-\bar{x}_{j}^{l}\rangle$ in \eqref{eq:opt1}, followed by rearrangement, yield
\begin{equation*}
	\begin{aligned}
		&\Delta D_{j,k}^{l}(\bar{x}_{j}^{l})\geq\alpha_{k}\langle \nabla_{1}f_{j}^{l}(x_{j,k}^{l},mv_{j,k}^{l})+m^{l}y_{j,k}^{l},x_{j,k+1}^{l}-\bar{x}_{j}^{l}\rangle\\		
		&+\alpha_{k}\lambda_{k}\langle \Delta Y_{j,k}^{l},x_{j,k+1}^{l}-\bar{x}_{j}^{l}\rangle+D_{\phi}(x_{j,k}^{l},x_{j,k+1}^{l})+\alpha_{k}\\		
		&\langle \nabla_{1}f_{j}^{l}(x_{j,k+1}^{l},mv_{j,k+1}^{l})+m^{l}y_{j,k+1}^{l},x_{j,k+1}^{l}-\bar{x}_{j}^{l}\rangle-\alpha_{k}\\			
		&\langle \nabla_{1}f_{j}^{l}(x_{j,k+1}^{l},mv_{j,k+1}^{l})+m^{l}y_{j,k+1}^{l},x_{j,k+1}^{l}-\bar{x}_{j}^{l}\rangle\\	
		&=\alpha_{k}\langle \nabla_{1}f_{j}^{l}(x_{j,k+1}^{l},mv_{j,k+1}^{l})+m^{l}y_{j,k+1}^{l},x_{j,k+1}^{l}-\bar{x}_{j}^{l}\rangle\!-\!\\
		&\alpha_{k}\langle\Delta Y_{j,k+1}^{l},x_{j,k+1}^{l}-\bar{x}_{j}^{l}\rangle+\alpha_{k}\lambda_{k}\langle \Delta Y_{j,k}^{l},x_{j,k}^{l}-\bar{x}_{j}^{l}\rangle+\alpha_{k}\lambda_{k}\\
		&\langle \Delta Y_{j,k}^{l},x_{j,k+1}^{l}-x_{j,k}^{l}\rangle+D_{\phi}(x_{j,k}^{l},x_{j,k+1}^{l}).
	\end{aligned}
\end{equation*}
By multiplying both sides of the above inequality by $\epsilon_{k}$, and summing over the indices $j$ from $1$ to $m^{l}$, $l$ from $1$ to $h$ and $k$ from $0$ to $t$, the following is derived:
\begin{equation}
	\begin{aligned}
		&\sum_{k=0}^{t}\epsilon_{k}\Delta D_{k}(\bar{\mathbf{x}})\geq\sum_{k=0}^{t}\epsilon_{k}\alpha_{k}\sum\limits_{l=1}^{h}\sum_{j=1}^{m^{l}}\langle \nabla_{1}f_{j}^{l}(x_{j,k+1}^{l},mv_{j,k+1}^{l})+\\		&m^{l}y_{j,k+1}^{l},x_{j,k+1}^{l}\!-\!\bar{x}_{j}^{l}\rangle\!+\!\sum_{k=0}^{t}\epsilon_{k}\alpha_{k}\lambda_{k}\sum\limits_{l=1}^{h}\sum_{j=1}^{m^{l}}\langle \Delta Y_{j,k}^{l},x_{j,k}^{l}\!-\!\bar{x}_{j}^{l}\rangle\\
		&-\sum_{k=0}^{t}\epsilon_{k}\alpha_{k}\sum\limits_{l=1}^{h}\sum_{j=1}^{m^{l}}\langle\Delta Y_{j,k+1}^{l},x_{j,k+1}^{l}-\bar{x}_{j}^{l}\rangle+\sum_{k=0}^{t}\epsilon_{k}\sum\limits_{l=1}^{h}\sum_{j=1}^{m^{l}}\\
		&D_{\phi}(x_{j,k}^{l},x_{j,k+1}^{l})+\sum_{k=0}^{t}\epsilon_{k}\alpha_{k}\lambda_{k}\sum\limits_{l=1}^{h}\sum_{j=1}^{m^{l}}\langle \Delta Y_{j,k}^{l},x_{j,k+1}^{l}-x_{j,k}^{l}\rangle\\
		&\geq -\epsilon_{t}\alpha_{t}\sum\limits_{l=1}^{h}\sum_{j=1}^{m^{l}}\langle\Delta Y_{j,t+1}^{l},x_{j,t+1}^{l}-\bar{x}_{j}^{l}\rangle+\sum_{k=0}^{t}\epsilon_{k}\alpha_{k}\sum\limits_{l=1}^{h}\\
		&\sum_{j=1}^{m^{l}}\langle \nabla_{1}f_{j}^{l}(x_{j,k+1}^{l},mv_{j,k+1}^{l})+m^{l}y_{j,k+1}^{l},x_{j,k+1}^{l}\!-\!\bar{x}_{j}^{l}\rangle+\\	&\underbrace{\sum_{k=0}^{t}\sum\limits_{l=1}^{h}\sum_{j=1}^{m^{l}}\big(\epsilon_{k}D_{\phi}(x_{j,k}^{l},x_{j,k\!+\!1}^{l})\!+\!\epsilon_{k}\alpha_{k}\lambda_{k}\langle \Delta Y_{j,k}^{l},x_{j,k\!+\!1}^{l}\!-\!x_{j,k}^{l}\rangle\big)}_{\triangleq R_{t}},
	\end{aligned}
	\label{eq:lem41}	
\end{equation}
where the second inequality follows from \eqref{eq:ep1} and the fact that $\Delta Y_{j,0}^{l}=0$ obtained by the initial condition. Further, based on the definitions of $G_{j}^{l}$, $F_{j}^{l}$ and $\hat{N}$, as well as Assumption \ref{asum2}, one can derive
\begin{equation*}
	\begin{aligned}
&\sum\limits_{l=1}^{h}\sum_{j=1}^{m^{l}}\langle \Delta Y_{j,k}^{l},x_{j,k\!+\!1}^{l}\!-\!x_{j,k}^{l}\rangle\!=\!\sum\limits_{l=1}^{h}\sum_{j=1}^{m^{l}}\Big\{\langle\nabla_{1}f_{j}^{l}(x_{j,k}^{l},mv_{j,k}^{l})\!+\!\\
&m^{l}y_{j,k}^{l}\!-\!G_{j}^{l}(x^{l}_{k},mv_{k}^{l}),x_{j,k\!+\!1}^{l}\!-\!x_{j,k}^{l}\rangle\!+\!\langle G_{j}^{l}(x^{l}_{k-1},mv_{k-1}^{l})-\\
&\nabla_{1}f_{j}^{l}(x_{j,k\!-\!1}^{l},mv_{j,k-1}^{l})\!-\!m^{l}y_{j,k-1}^{l},x_{j,k\!+\!1}^{l}\!-\!x_{j,k}^{l}\rangle\!+\!\langle G_{j}^{l}(x^{l}_{k},\\
&mv_{k}^{l})\!-\!F_{j}^{l}(\mathbf{x}_{k}),x_{j,k\!+\!1}^{l}\!-\!x_{j,k}^{l}\rangle+\langle F_{j}^{l}(\mathbf{x}_{k-1})-G_{j}^{l}(x^{l}_{k-1},\\
&mv_{k-1}^{l}),x_{j,k\!+\!1}^{l}\!-\!x_{j,k}^{l}\rangle+\langle F_{j}^{l}(\mathbf{x}_{k})-F_{j}^{l}(\mathbf{x}_{k-1}),x_{j,k\!+\!1}^{l}\!-\!x_{j,k}^{l}\rangle\Big\}\\
\end{aligned}
\end{equation*}
\begin{equation}
\begin{aligned}	
&\geq-\sum\limits_{l=1}^{h}\sum_{j=1}^{m^{l}}\big(\lVert m^{l}y_{j,k}^{l}\!-\!\sum_{i=1}^{m^{l}}\nabla_{2}f_{i}^{l}(x_{i,k}^{l},mv_{i,k}^{l})\rVert_{*}+\lVert m^{l}y_{j,k-1}^{l}\\
&\!-\!\sum_{i=1}^{m^{l}}\nabla_{2}f_{i}^{l}(x_{i,k-1}^{l},mv_{i,k-1}^{l})\rVert_{*}\big)\lVert x_{j,k\!+\!1}^{l}\!-\!x_{j,k}^{l}\rVert-\sum\limits_{l=1}^{h}\sum_{j=1}^{m^{l}}\\
&\big(\lVert G_{j}^{l}(x^{l}_{k},mv_{k}^{l})\!-\!F_{j}^{l}(\mathbf{x}_{k})\rVert_{*}\!+\!\lVert F_{j}^{l}(\mathbf{x}_{k\!-\!1})\!-\!G_{j}^{l}(x^{l}_{k\!-\!1},mv_{k\!-\!1}^{l}) \rVert_{*}\big)\\
&\lVert x_{j,k\!+\!1}^{l}\!-\!x_{j,k}^{l}\rVert+\langle F(\mathbf{x}_{k})-F(\mathbf{x}_{k-1}),\mathbf{x}_{k+1}\!-\!\mathbf{x}_{k}\rangle\\
&\geq\!-\!2\hat{N}\sum\limits_{l=1}^{h}\sum_{j=1}^{m^{l}}\big(\lVert m^{l}y_{j,k}^{l}\!-\!\sum_{i=1}^{m^{l}}\nabla_{2}f_{i}^{l}(x_{i,k}^{l},mv_{i,k}^{l})\rVert_{*}\!+\!\lVert m^{l}y_{j,k\!-\!1}^{l}\\
&\!-\!\sum_{i=1}^{m^{l}}\nabla_{2}f_{i}^{l}(x_{i,k-1}^{l},mv_{i,k-1}^{l})\rVert_{*}+\bar{L}\sum_{i=1}^{m^{l}}\big(\lVert mv_{i,k}^{l}-\psi(\mathbf{x}_{k})\rVert\\
&+\lVert mv_{i,k-1}^{l}-\psi(\mathbf{x}_{k-1})\rVert\big)\big)-L\lVert\mathbf{x}_{k}-\mathbf{x}_{k-1}\rVert\lVert\mathbf{x}_{k+1}\!-\!\mathbf{x}_{k}\rVert\\
&=-\delta_{k}-L\lVert\mathbf{x}_{k}-\mathbf{x}_{k-1}\rVert\lVert\mathbf{x}_{k+1}\!-\!\mathbf{x}_{k}\rVert,
\end{aligned}
\label{eq:lem42}	
\end{equation}
where the last equality is a result of invoking \eqref{eq:delta}. Then, by substituting the above result into $R_{t}$ and applying the relationship \eqref{eq:relat} between the Bregman divergence and the Euclidean distance, one obtains under the initial condition $\mathbf{x}_{0}=\mathbf{x}_{-1}$,
\begin{equation*}
	\begin{aligned}
&R_{t}\geq\sum_{k=0}^{t}\epsilon_{k}\bar{D}_{\phi}(\mathbf{x}_{k},\mathbf{x}_{k+1})\!-\!\sum_{k=0}^{t}\epsilon_{k}\alpha_{k}\lambda_{k}\delta_{k}-\sum_{k=0}^{t}\epsilon_{k}\alpha_{k}\lambda_{k}L\\
&\lVert\mathbf{x}_{k}-\mathbf{x}_{k-1}\rVert\lVert\mathbf{x}_{k+1}\!-\!\mathbf{x}_{k}\rVert	\\
&\geq\frac{1}{2}\sum_{k=0}^{t}\epsilon_{k}\bar{D}_{\phi}(\mathbf{x}_{k},\mathbf{x}_{k\!+\!1})\!-\!\sum_{k=0}^{t}\epsilon_{k}\alpha_{k}\lambda_{k}\delta_{k}\!+\!\sum_{k=0}^{t}\big(\frac{\epsilon_{k}}{8}\lVert\mathbf{x}_{k\!+\!1}\!-\!\mathbf{x}_{k}\rVert^{2}\\
&+\frac{\epsilon_{k-1}}{8}\lVert\mathbf{x}_{k}\!-\!\mathbf{x}_{k-1}\rVert^{2}-\epsilon_{k}\alpha_{k}\lambda_{k}L\lVert\mathbf{x}_{k}-\mathbf{x}_{k-1}\rVert\lVert\mathbf{x}_{k+1}\!-\!\mathbf{x}_{k}\rVert\big)\\
&+\frac{\epsilon_{t}}{8}\lVert\mathbf{x}_{t\!+\!1}\!-\!\mathbf{x}_{t}\rVert^{2}.	
\end{aligned}	
\end{equation*}
Notably, the term $\frac{\epsilon_{k}}{8}\lVert\mathbf{x}_{k\!+\!1}\!-\!\mathbf{x}_{k}\rVert^{2}+\frac{\epsilon_{k-1}}{8}\lVert\mathbf{x}_{k}\!-\!\mathbf{x}_{k-1}\rVert^{2}-\alpha_{k}\lambda_{k}\epsilon_{k}L\lVert\mathbf{x}_{k}-\mathbf{x}_{k-1}\rVert\lVert\mathbf{x}_{k+1}\!-\!\mathbf{x}_{k}\rVert\geq0$ holds by \eqref{eq:ep2} and the Young's inequality. Thus, one has
\begin{equation*}
	\begin{aligned}
R_{t}\geq\frac{1}{2}\sum_{k=0}^{t}\epsilon_{k}\bar{D}_{\phi}(\mathbf{x}_{k},\mathbf{x}_{k+1})\!-\!\sum_{k=0}^{t}\epsilon_{k}\alpha_{k}\lambda_{k}\delta_{k}\!+\!\frac{\epsilon_{t}}{8}\lVert\mathbf{x}_{t\!+\!1}\!-\!\mathbf{x}_{t}\rVert^{2},		\end{aligned}
\end{equation*}		
which is brought back to \eqref{eq:lem41}, yielding
\begin{equation}
	\begin{aligned}
		&\sum_{k=0}^{t}\epsilon_{k}\Delta D_{k}(\bar{\mathbf{x}})\geq\frac{1}{2}\sum_{k=0}^{t}\epsilon_{k}\bar{D}_{\phi}(\mathbf{x}_{k},\mathbf{x}_{k+1})\!-\!\sum_{k=0}^{t}\epsilon_{k}\alpha_{k}\lambda_{k}\delta_{k}+\\
		&\sum_{k=0}^{t}\epsilon_{k}\alpha_{k}\sum\limits_{l=1}^{h}\sum_{j=1}^{m^{l}}\langle \nabla_{1}f_{j}^{l}(x_{j,k+1}^{l},mv_{j,k+1}^{l})+m^{l}y_{j,k+1}^{l},x_{j,k+1}^{l}\\	&\!-\!\bar{x}_{j}^{l}\rangle+\frac{\epsilon_{t}}{8}\lVert\mathbf{x}_{t\!+\!1}\!-\!\mathbf{x}_{t}\rVert^{2}-\epsilon_{t}\alpha_{t}\sum\limits_{l=1}^{h}\sum_{j=1}^{m^{l}}\langle\Delta Y_{j,t+1}^{l},x_{j,t+1}^{l}-\bar{x}_{j}^{l}\rangle.
\end{aligned}
\label{eq:lem44}	
\end{equation}
Proceeding as in the derivation of \eqref{eq:lem42} and using $G_{j}^{l}$, $F_{j}^{l}$, $\hat{N}$ and Assumption \ref{asum2}, one obtains:
\begin{equation*}
	\begin{aligned}
		&\!-\!\sum\limits_{l=1}^{h}\sum_{j=1}^{m^{l}}\langle\Delta Y_{j,t\!+\!1}^{l},x_{j,t\!+\!1}^{l}\!-\!\bar{x}_{j}^{l}\rangle\!=\!\sum\limits_{l=1}^{h}\sum_{j=1}^{m^{l}}\langle\Delta Y_{j,t\!+\!1}^{l},\bar{x}_{j}^{l}\!-\!x_{j,t\!+\!1}^{l}\rangle\\
		&\geq-\delta_{t+1}-L\lVert\mathbf{x}_{t+1}-\mathbf{x}_{t}\rVert\lVert\mathbf{x}_{t+1}\!-\!\bar{\mathbf{x}}\rVert.
\end{aligned}	
\end{equation*}
Then, the following inequality can be derived
\begin{equation*}
	\begin{aligned}		&\frac{\epsilon_{t}}{8}\lVert\mathbf{x}_{t\!+\!1}\!-\!\mathbf{x}_{t}\rVert^{2}-\epsilon_{t}\alpha_{t}\sum\limits_{l=1}^{h}\sum_{j=1}^{m^{l}}\langle\Delta Y_{j,t+1}^{l},x_{j,t+1}^{l}-\bar{x}_{j}^{l}\rangle\\	
   &\geq\frac{\epsilon_{t}}{8}\lVert\mathbf{x}_{t\!+\!1}\!-\!\mathbf{x}_{t}\rVert^{2}\!-\!\epsilon_{t}\alpha_{t}\delta_{t\!+\!1}\!-\!\epsilon_{t}\alpha_{t}L\lVert\mathbf{x}_{t\!+\!1}\!-\!\mathbf{x}_{t}\rVert\lVert\mathbf{x}_{t+1}\!-\!\bar{\mathbf{x}}\rVert\\	
		&\geq-\!\epsilon_{t}\alpha_{t}\delta_{t\!+\!1}-2\epsilon_{t}\alpha_{t}^{2}L^{2}\lVert\mathbf{x}_{t+1}\!-\!\bar{\mathbf{x}}\rVert^{2},
	\end{aligned}	
\end{equation*}
where the last inequality is established by $\frac{1}{8}\lVert\mathbf{x}_{t\!+\!1}\!-\!\mathbf{x}_{t}\rVert^{2}\!-\!\alpha_{t}L\lVert\mathbf{x}_{t\!+\!1}\!-\!\mathbf{x}_{t}\rVert\lVert\mathbf{x}_{t+1}\!-\!\bar{\mathbf{x}}\rVert=(\frac{1}{2\sqrt{2}}\lVert\mathbf{x}_{t\!+\!1}\!-\!\mathbf{x}_{t}\rVert-\sqrt{2}\alpha_{t}L\lVert\mathbf{x}_{t+1}\!-\!\bar{\mathbf{x}}\rVert)^{2}-2\alpha_{t}^{2}L^{2}\lVert\mathbf{x}_{t+1}\!-\!\bar{\mathbf{x}}\rVert^2\geq-2\alpha_{t}^{2}L^{2}\lVert\mathbf{x}_{t+1}\!-\!\bar{\mathbf{x}}\rVert^2$. By substituting the above results into the corresponding terms in \eqref{eq:lem44}, it follows that
\begin{equation*}
	\begin{aligned}
		&\sum_{k=0}^{t}\epsilon_{k}\Delta D_{k}(\bar{\mathbf{x}})\geq\frac{1}{2}\sum_{k=0}^{t}\epsilon_{k}\bar{D}_{\phi}(\mathbf{x}_{k},\mathbf{x}_{k+1})\!-\!\sum_{k=0}^{t}\epsilon_{k}\alpha_{k}\lambda_{k}\delta_{k}+\\
		&\sum_{k=0}^{t}\epsilon_{k}\alpha_{k}\sum\limits_{l=1}^{h}\sum_{j=1}^{m^{l}}\langle \nabla_{1}f_{j}^{l}(x_{j,k+1}^{l},mv_{j,k+1}^{l})+m^{l}y_{j,k+1}^{l},x_{j,k+1}^{l}\\
&\!-\!\bar{x}_{j}^{l}\rangle\!-\!\epsilon_{t}\alpha_{t}\delta_{t\!+\!1}-2\epsilon_{t}\alpha_{t}^{2}L^{2}\lVert\mathbf{x}_{t+1}\!-\!\bar{\mathbf{x}}\rVert^{2}.
	\end{aligned}
\end{equation*}
Incorporating \eqref{eq:dto} into the above inequality and rearranging gives the desired result. $\blacksquare$

\subsection{Main Results}
In this subsection, the convergence property of Algorithm \ref{alg:alg1} is explored with appropriate step sizes in Theorem \ref{th:thm1}, building on Lemma \ref{lem4}. Subsequently, its convergence rate is identified in Theorem \ref{th:thm2} by combining Lemma \ref{lem2} and Theorem \ref{th:thm1}.

\begin{thm}\label{th:thm1}
Under Assumptions \ref{asum1}-\ref{asum4}, consider Algorithm \ref{alg:alg1}. Let the parameters $\alpha_{k}$, $\lambda_{k}$ and $\epsilon_{k}$ satisfy \eqref{eq:ep1} and \eqref{eq:ep2}. Further, suppose that the parameters $\alpha_{k}$ and $\epsilon_{k}$ satisfy
\begin{equation}
\epsilon_{k}\leq\epsilon_{k-1}(2\mu\alpha_{k-1}+1)\text{ and }8L^{2}\alpha_{k}^{2}\leq1.
\label{eq:alep}
\end{equation}
Then
\begin{equation}
	\begin{aligned} &\epsilon_{t}(2\mu\alpha_{t}+\frac{1}{2})\bar{D}_{\phi}(\mathbf{x}_{t+1},\mathbf{x}^{*})+\frac{1}{2}\sum_{k=0}^{t}\epsilon_{k}\bar{D}_{\phi}(\mathbf{x}_{k},\mathbf{x}_{k+1})\\ &\leq\epsilon_{0}\bar{D}_{\phi}(\mathbf{x}_{0},\mathbf{x}^{*})\!+\!2\sum_{k=0}^{t}\epsilon_{k}\alpha_{k}\delta_{k+1}+\epsilon_{0}\alpha_{0}\lambda_{0}\delta_{0}
	\end{aligned}
	\label{eq:th1}	
\end{equation}	
holds for all $t>0$.
\end{thm}

\noindent\textbf{Proof.} By introducing $G_{j}^{l}(x^{l}_{k+1},mv_{k+1}^{l})$ and $F_{j}^{l}(\mathbf{x}_{k+1})$ and applying Assumption \ref{asum2}, one has
\begin{equation*}
	\begin{aligned}
&\langle \nabla_{1}f_{j}^{l}(x_{j,k\!+\!1}^{l},mv_{j,k\!+\!1}^{l})\!+\!m^{l}y_{j,k\!+\!1}^{l},x_{j,k\!+\!1}^{l}\!-\!\bar{x}_{j}^{l}\rangle\!=\!\langle \big(\nabla_{1}f_{j}^{l}(x_{j,k\!+\!1}^{l},\\
&mv_{j,k\!+\!1}^{l})\!+\!m^{l}y_{j,k\!+\!1}^{l}\big)\!-\!G_{j}^{l}(x^{l}_{k\!+\!1},mv_{k\!+\!1}^{l}),x_{j,k\!+\!1}^{l}\!-\!\bar{x}_{j}^{l}\rangle\!+\!\langle G_{j}^{l}(x^{l}_{k\!+\!1},\\
&mv_{k+1}^{l})-F_{j}^{l}(\mathbf{x}_{k+1}),x_{j,k+1}^{l}\!-\!\bar{x}_{j}^{l}\rangle+\langle F_{j}^{l}(\mathbf{x}_{k+1}),x_{j,k+1}^{l}\!-\!\bar{x}_{j}^{l}\rangle\\
&\geq\langle F_{j}^{l}(\mathbf{x}_{k+1}),x_{j,k+1}^{l}\!-\!\bar{x}_{j}^{l}\rangle-2\hat{N}\lVert m^{l}y_{j,k\!+\!1}^{l}\!-\!\sum_{i=1}^{m^{l}}\nabla_{2}f_{i}^{l}(x_{i,k\!+\!1}^{l},\\
&mv_{i,k\!+\!1}^{l})\rVert_{*}-2\hat{N}\bar{L}\sum_{i=1}^{m^{l}}\lVert mv_{i,k+1}^{l}-\psi(\mathbf{x}_{k+1})\rVert.
\end{aligned}
\end{equation*}
Further, setting $\bar{x}_{j}^{l}=x_{j}^{l,*}$, and employing Assumption \ref{asum3} and the definition \eqref{eq:delta} of $\delta_{k}$ lead to
\begin{equation*}
	\begin{aligned}
		&\sum\limits_{l=1}^{h}\sum_{j=1}^{m^{l}}\langle \nabla_{1}f_{j}^{l}(x_{j,k\!+\!1}^{l},mv_{j,k\!+\!1}^{l})\!+\!m^{l}y_{j,k\!+\!1}^{l},x_{j,k\!+\!1}^{l}\!-\!x_{j}^{l,*}\rangle\\
		&\geq\langle F(\mathbf{x}_{k\!+\!1}),\mathbf{x}_{k+1}\!-\!\mathbf{x}^{*}\rangle\!-\!\delta_{k\!+\!1} \geq2\mu\bar{D}_{\phi}(\mathbf{x}_{k\!+\!1},\mathbf{x}^{*})\!-\!\delta_{k+1},
	\end{aligned}
\end{equation*}
which is substituted back into \eqref{eq:lem4} with $\bar{x}_{j}^{l}=x_{j}^{l,*}$, and invoking \eqref{eq:ep1}, resulting in
\begin{equation}
	\begin{aligned} &\sum_{k=0}^{t}\epsilon_{k}(2\mu\alpha_{k}+1)\bar{D}_{\phi}(\mathbf{x}_{k+1},\mathbf{x}^{*})+\!\frac{1}{2}\sum_{k=0}^{t}\epsilon_{k}\bar{D}_{\phi}(\mathbf{x}_{k},\mathbf{x}_{k+1})\\
		&\!-\!2\epsilon_{t}\alpha_{t}^{2}L^{2}\lVert\mathbf{x}_{t+1}\!-\!\mathbf{x}^{*}\rVert^{2}\\		&\leq\sum_{k=0}^{t}\epsilon_{k}\bar{D}_{\phi}(\mathbf{x}_{k},\mathbf{x}^{*})\!+\!\sum_{k\!=\!0}^{t}\epsilon_{k}\alpha_{k}\lambda_{k}\delta_{k}\!+\!\epsilon_{t}\alpha_{t}\delta_{t\!+\!1}\!+\!\sum_{k\!=\!0}^{t}\epsilon_{k}\alpha_{k}\delta_{k\!+\!1}\\
		&=\sum_{k=0}^{t}\epsilon_{k}\bar{D}_{\phi}(\mathbf{x}_{k},\mathbf{x}^{*})\!+\!2\sum_{k=0}^{t}\epsilon_{k}\alpha_{k}\delta_{k+1}+\epsilon_{0}\alpha_{0}\lambda_{0}\delta_{0}.
	\end{aligned}
\label{eq:th11}	
\end{equation}	
Given that $\epsilon_{k}\leq\epsilon_{k-1}(2\mu\alpha_{k-1}+1)$ indicated in \eqref{eq:alep}, it follows that
\begin{equation*}
	\begin{aligned} &\sum_{k=0}^{t}\epsilon_{k}(2\mu\alpha_{k}+1)\bar{D}_{\phi}(\mathbf{x}_{k+1},\mathbf{x}^{*})\!-\!2\epsilon_{t}\alpha_{t}^{2}L^{2}\lVert\mathbf{x}_{t+1}\!-\!\mathbf{x}^{*}\rVert^{2}\\		&\geq\epsilon_{t}(2\mu\alpha_{t}+\frac{1}{2})\bar{D}_{\phi}(\mathbf{x}_{t+1},\mathbf{x}^{*})+\sum_{k=0}^{t-1}\epsilon_{k+1}\bar{D}_{\phi}(\mathbf{x}_{k+1},\mathbf{x}^{*})+\\	&\frac{\epsilon_{t}}{2}\big(\bar{D}_{\phi}(\mathbf{x}_{t+1},\mathbf{x}^{*})-4\alpha_{t}^{2}L^{2}\lVert\mathbf{x}_{t+1}\!-\!\mathbf{x}^{*}\rVert^{2}\big)\\
&\geq\epsilon_{t}(2\mu\alpha_{t}+\frac{1}{2})\bar{D}_{\phi}(\mathbf{x}_{t+1},\mathbf{x}^{*})+\sum_{k=0}^{t-1}\epsilon_{k+1}\bar{D}_{\phi}(\mathbf{x}_{k+1},\mathbf{x}^{*}),
	\end{aligned}
\end{equation*}	
where the last inequality is derived from \eqref{eq:relat} and $8L^{2}\alpha_{k}^{2}\leq1$ indicated in \eqref{eq:alep}. The desired result can be obtained by combining the above result with \eqref{eq:th11}. $\blacksquare$

Next, the convergence rate of Algorithm \ref{alg:alg1} is shown under specific parameter selections.

\begin{thm}\label{th:thm2}
Under Assumptions \ref{asum1}-\ref{asum4}, consider Algorithm \ref{alg:alg1}. Let
\begin{equation}
	\begin{aligned}
	&p\!=\!\frac{5L}{\mu}, \alpha_{k}\!=\!\frac{1}{\mu(k\!+\!p\!-\!1)},\epsilon_{k}\!=\!(k\!+\!p\!+\!1)(k\!+\!p), \lambda_{k}\!=\!\frac{\epsilon_{k\!-\!1}\alpha_{k\!-\!1}}{\epsilon_{k}\alpha_{k}}.
	\end{aligned}
\label{eq:par}
\end{equation}	
Then
\begin{equation}
	\begin{aligned} &\bar{D}_{\phi}(\mathbf{x}_{t+1},\mathbf{x}^{*})\!\leq\frac{2(p+1)p}{\!(t\!+\!p\!+\!1)(t\!+\!p)}\bar{D}_{\phi}(\mathbf{x}_{0},\mathbf{x}^{*})\!\\
		&+\!\frac{8(q_{1}+q_{2}(t+1))}{\mu^{2}\!(t\!+\!p\!+\!1)(t\!+\!p)}+\frac{2p(p-1)p_{\delta}}{\mu\!(p\!-\!2)(t\!+\!p\!+\!1)(t\!+\!p)}
	\end{aligned}
\label{eq:th2}
\end{equation}	
holds for all $t\geq0$, where $p_{\delta}\triangleq4\hat{N}m(\bar{m}\tilde{W}_{y}+\bar{L}m\tilde{W}_{v})$, $q_{1}$ and $q_{2}$ are defined in \eqref{eq:th27}.	
\end{thm}

\noindent\textbf{Proof.} It follows from \eqref{eq:vx2} and \eqref{eq:yf2} that, for all $k\geq0$, $\delta_{k+1}$ is bounded as
\begin{equation}
	\begin{aligned}
		&\delta_{k}\!\leq4\hat{N}m(\bar{m}\tilde{W}_{y}+\bar{L}m\tilde{W}_{v})=p_{\delta}.
	\end{aligned}
\label{eq:th21}
\end{equation}	
Similarly to \citet{Kot22}, it can be verified that \eqref{eq:ep1}, \eqref{eq:ep2}, and \eqref{eq:alep} hold under the parameter selection \eqref{eq:par}. Additionally, for $p\geq5$, straightforward calculations yield
\begin{equation*}
	\begin{aligned}
&\epsilon_{t}(2\mu\alpha_{t}+\frac{1}{2})>\frac{\epsilon_{t}}{2}=\frac{\!(t\!+\!p\!+\!1)(t\!+\!p)}{2},\\
&\epsilon_{k}\alpha_{k}=\frac{\!(k\!+\!p\!+\!1)(k\!+\!p)}{\mu(k\!+\!p-1)}\leq\frac{2}{\mu^{2}\alpha_{k}}.
	\end{aligned}
\end{equation*}	
Combining the above relations with \eqref{eq:th1} and \eqref{eq:th21}, and using $\bar{D}_{\phi}(\mathbf{x}_{k},\mathbf{x}_{k\!+\!1})\geq0$, yields
\begin{equation}
	\begin{aligned} \frac{\!(t\!+\!p\!+\!1)(t\!+\!p)}{2}&\bar{D}_{\phi}(\mathbf{x}_{t\!+\!1},\mathbf{x}^{*})\!\leq\epsilon_{0}\bar{D}_{\phi}(\mathbf{x}_{0},\mathbf{x}^{*})\\
		&\!+\!\frac{4}{\mu^{2}}\sum_{k=0}^{t}\frac{\delta_{k\!+\!1}}{\alpha_{k}}+\epsilon_{0}\alpha_{0}\lambda_{0}p_{\delta}.
	\end{aligned}
\label{eq:th22}	
\end{equation}	
Next, one proceeds to analyze an upper bound on $\sum_{k=0}^{t}\frac{\delta_{k\!+\!1}}{\alpha_{k}}$.
For any $t\geq0$, the following estimates hold:
\begin{equation*}
	\begin{aligned}	&\sum_{k=0}^{t}k\eta^{k}=\eta\big(\sum_{k=0}^{t}k\eta^{k-1}\big)=\eta\big(\sum_{k=0}^{t}\eta^{k}\big)'=\eta\big(\frac{1-\eta^{t+1}}{1-\eta}\big)'\\
		&~~~~~~=\frac{\eta(1-(t+1)\eta^{t}(1-\eta)-\eta^{t+1})}{(1-\eta)^{2}}\leq\frac{\eta}{(1-\eta)^2}\\ &\sum_{k=0}^{t}k(k\!-\!1)\eta^{k\!-\!2}\!=\!\big(\frac{1\!-\!\eta^{t\!+\!1}}{1\!-\!\eta}\big)''\!=\!(\frac{(1\!-\!(t\!+\!1)\eta^{t}(1\!-\!\eta)\!-\!\eta^{t\!+\!1})}{(1\!-\!\eta)^{2}})'\\		
&~~~~~~\!=\!\frac{2\!-\!2(t\!+\!1)\eta^{t}(1\!-\!\eta)-2\eta^{t\!+\!1}\!-\!(t\!+\!1)t\eta^{t\!-\!1}(1\!-\!\eta)^{2}}{(1\!-\!\eta)^{3}}\\
		&~~~~~~\leq\frac{3}{(1-\eta)^3}\\	
		&\sum_{k=0}^{t}\frac{\eta^{k}}{\alpha_{k}}=\sum_{k=0}^{t}\mu(k+p-1)\eta^{k}\leq\frac{\mu\eta}{(1-\eta)^2}+\frac{(p-1)\mu}{1-\eta}\\
	\end{aligned}
\end{equation*}
\begin{equation}
\begin{aligned}
		&\sum_{k=0}^{t}\frac{1}{\alpha_{k}}\sum_{s=1}^{k}\eta^{k-s}\alpha_{s-1}=\sum_{k=0}^{t}\sum_{s=1}^{k}\eta^{k-s}\frac{\alpha_{s-1}}{\alpha_{k}}\\
		&~~~~~~\!=\!\sum_{k=0}^{t}\sum_{s=1}^{k}\eta^{k-s}\frac{k+p-1}{s-2+p}\!=\!\sum_{k=0}^{t}\sum_{s=1}^{k}\eta^{k-s}(1+\frac{k-s+1}{s-2+p})\\
		&~~~~~~\stackrel{(e\triangleq k-s)}{\leq}\sum_{k=0}^{t}\sum_{e=0}^{k-1}\eta^{e}(e+2)\leq\frac{(2-\eta)(t+1)}{(1-\eta)^{2}}\\	
	&\sum_{k=0}^{t}\frac{1}{\alpha_{k}}\sum_{s=1}^{k-1}\eta^{k-1-s}\alpha_{s-1}=\sum_{k=0}^{t}\sum_{s=1}^{k-1}\eta^{k-1-s}(1+\frac{k-s+1}{s+p-2})\\
		&~~~~~~\leq\frac{(3-2\eta)(t+1)}{(1-\eta)^{2}}\\		
&\sum_{k=0}^{t}\frac{1}{\alpha_{k}}\sum_{s=1}^{k-2}\eta^{k-2-s}\alpha_{s-1}\leq\frac{(4-3\eta)(t+1)}{(1-\eta)^{2}}\\
&\sum_{k=0}^{t}\frac{1}{\alpha_{k}}\sum_{s=1}^{k}\eta^{k-s}\alpha_{s-2}\leq\frac{(3-2\eta)(t+1)}{(1-\eta)^{2}}\\
&\sum_{k=0}^{t}\frac{k\eta^{k\!-\!2}}{\alpha_{k}}\!=\!\sum_{k=0}^{t}\mu(k\!+\!p\!-\!1)k\eta^{k\!-\!2}\!=\!\sum_{k=0}^{t}\mu((k\!-\!1)k\!+\!pk)\eta^{k\!-\!2}\\
&~~~~~~\leq\frac{3\mu}{(1-\eta)^3}+\frac{\mu p}{(1-\eta)^2\eta}\\	
&\sum_{k=0}^{t}\frac{k\eta^{k\!-\!3}}{\alpha_{k}}\leq\frac{3\mu}{(1-\eta)^3\eta}+\frac{\mu p}{(1-\eta)^2\eta^2}\\
&\sum_{k=0}^{t}\frac{1}{\alpha_{k}}\sum_{s=1}^{k}\eta^{k-s}\sum_{r=1}^{s-2}\eta^{s-2-r}\alpha_{r-1}\\
&~~~~~~=\sum_{k=0}^{t}\frac{\sum_{s=1}^{k}\eta^{k-s}\alpha_{s-2}}{\alpha_{k}}\frac{\sum_{r=1}^{s-2}\eta^{s-2-r}\alpha_{r-1}}{\alpha_{s-2}}\\
&~~~~~~=\sum_{k=0}^{t}\big(\sum_{s=1}^{k}\eta^{k\!-\!s}(1\!+\!\frac{k\!-\!s\!+\!2}{s\!-\!3\!+\!p})\big)\big(\sum_{r=1}^{s\!-\!2}\eta^{s\!-\!2\!-\!r}(1\!+\!\frac{s\!-\!r\!-\!1}{r\!-\!2\!+\!p})\big)\\
&~~~~~~\leq\sum_{k=0}^{t}(\frac{3-2\eta}{(1-\eta)^{2}})(\frac{2-\eta}{(1-\eta)^{2}})=\frac{(3-2\eta)(2-\eta)(t+1)}{(1-\eta)^{4}}\\
&\sum_{k=0}^{t}\frac{1}{\alpha_{k}}\sum_{s=1}^{k-1}\eta^{k-1-s}\sum_{r=1}^{s-2}\eta^{s-2-r}\alpha_{r-1}\leq\frac{(4-3\eta)(2-\eta)(t+1)}{(1-\eta)^{4}}.
\end{aligned}
\label{eq:th23}	
\end{equation}	
Also, note that
\begin{equation}
		\lambda_{k}=\frac{\epsilon_{k\!-\!1}\alpha_{k\!-\!1}}{\epsilon_{k}\alpha_{k}}=\frac{(k+p-1)^2}{(k+p-2)(k+p+1)}\leq\frac{3}{2}
\label{eq:th24}	
\end{equation}	
holds with $p\geq4$ for any $k\geq0$. Further, by employing \eqref{eq:th23} and \eqref{eq:th24}, one has
\begin{equation*}
	\begin{aligned}
&\sum_{k=0}^{t}\frac{mEv^{l}_{j,k+1}}{\alpha_{k}}\leq m^{2}\beta\hat{N}(\frac{\mu\eta}{(1-\eta)^2}\!+\!\frac{(p-1)\mu}{1-\eta})\\
&~~~~~~+4m^{2}\tilde{W}\beta\frac{(2-\eta)(t+1)}{(1-\eta)^{2}}+8m\tilde{W}(t+1),\\
\end{aligned}
\end{equation*}
\begin{equation}
\begin{aligned}	
&\sum_{k=0}^{t}\frac{mEv^l_{j,k}}{\alpha_{k}}\leq m^{2}\beta\hat{N}(\frac{\mu}{(1-\eta)^2}+\frac{(p-1)\mu}{(1-\eta)\eta})\\
&~~~~~~+4m^{2}\tilde{W}\beta\frac{(3-2\eta)(t+1)}{(1-\eta)^{2}}+16m\tilde{W}(t+1),
\end{aligned}
\label{eq:th25}	
\end{equation}	
where the last inequality is based on $\frac{\alpha_{k-1}}{\alpha_{k}}=\frac{k+p-1}{k+p-2}\leq2$. Similarly, for the gradient tracking error, by leveraging \eqref{eq:th23} and \eqref{eq:th24}, the following can be derived,
\begin{equation}
	\begin{aligned}
		&\sum_{k=0}^{t}\frac{m^lEy^l_{j,k+1}}{\alpha_{k}}\!\leq\!(m^{l})^{2}N_{f}\beta(\frac{\mu\eta}{(1\!-\!\eta)^2}\!+\!\frac{(p-1)\mu}{1\!-\!\eta})\!+\!\\		
		&~~~~~~4(m\!+\!1)(m^{l})^{2}\hat{L}\tilde{W}\beta\frac{(2\!-\!\eta)(t\!+\!1)}{(1\!-\!\eta)^{2}}+16m(m^{l})^{2}\hat{L}\tilde{W}\beta\\	
&~~~~~~\frac{(3-2\eta)(t+1)}{(1-\eta)^{2}}+2(mm^{l}\beta)^{2}\hat{N}\hat{L}\big(\frac{3\mu}{(1-\eta)^3}+\frac{\mu p}{(1-\eta)^2\eta}\big)\\
		&~~~~~~+4(m^{l})^{2}\hat{L}\tilde{W}\beta^{2}\frac{(3-2\eta)(2-\eta)(t+1)}{(1-\eta)^{4}}+8m^{l}(m+1)\\
&~~~~~~\hat{L}\tilde{W}(t+1)+64m^{l}m\hat{L}\tilde{W}(t+1)+4m^2m^{l}\hat{L}\hat{N}\beta\\
		&~~~~~~\big(\frac{\mu}{(1-\eta)^2}+\frac{(p-1)\mu}{(1-\eta)\eta}\big)+16m^2m^{l}\hat{L}\tilde{W}\beta\frac{(3-2\eta)(t+1)}{(1-\eta)^{2}},\\
		&\sum_{k=0}^{t}\frac{m^lEy^l_{j,k}}{\alpha_{k}}\!\leq\!(m^{l})^{2}N_{f}\beta(\frac{\mu}{(1\!-\!\eta)^2}\!+\!\frac{(p-1)\mu}{(1\!-\!\eta)\eta})\!+\!4(m\!+\!1)\\
		&~~~~ (m^{l})^{2}\hat{L}\tilde{W}\beta\frac{(3\!-\!2\eta)(t\!+\!1)}{(1\!-\!\eta)^{2}}+16m(m^{l})^{2}\hat{L}\tilde{W}\beta\\
		&~~~~\frac{(4-3\eta)(t+1)}{(1-\eta)^{2}}+2(mm^{l}\beta)^{2}\hat{N}\hat{L}\big(\frac{3\mu}{(1-\eta)^3\eta}+\\
		&~~~~\frac{\mu p}{(1-\eta)^2\eta^2}\big)+4(m^{l})^{2}\hat{L}\tilde{W}\beta^{2}\frac{(4-3\eta)(2-\eta)(t+1)}{(1-\eta)^{4}}+\\		
		&~~~~16m^{l}(m\!+\!1)\hat{L}\tilde{W}(t\!+\!1)\!+\!96m^{l}m\hat{L}\tilde{W}(t\!+\!1)\!+\!4m^2m^{l}\hat{L}\hat{N}\beta\\	&~~~~\big(\frac{\mu}{(1\!-\!\eta)^2\eta}\!+\!\frac{(p\!-\!1)\mu}{(1\!-\!\eta)\eta^{2}}\big)\!+\!16m^2m^{l}\hat{L}\tilde{W}\beta\frac{(4\!-\!3\eta)(t\!+\!1)}{(1\!-\!\eta)^{2}}.\\	
	\end{aligned}
	\label{eq:th26}
\end{equation}
Combining \eqref{eq:th25} and \eqref{eq:th26}, one proceeds to derive
	\begin{equation*}
\begin{aligned}
&\sum_{k=0}^{t}\frac{\delta_{k\!+\!1}}{\alpha_{k}}\leq2\hat{N}\bar{L}m\Big\{\Big[m^{2}\beta\hat{N}(1\!+\!\eta)(\frac{\mu}{(1\!-\!\eta)^2}\!+\!\frac{(p-1)\mu}{(1\!-\!\eta)\eta})\Big]\\
&\!+\!\Big[4m^{2}\tilde{W}\beta\frac{(5\!-\!3\eta)}{(1\!-\!\eta)^{2}}\!+\!24m\tilde{W}\Big](t\!+\!1)\Big\}\!+\!2\hat{N}m\Big\{\Big[(\bar{m})^{2}N_{f}\beta\\
&(1\!+\!\eta)(\frac{\mu}{(1\!-\!\eta)^2}\!+\!\frac{(p\!-\!1)\mu}{(1\!-\!\eta)\eta})\!+\!2(m\bar{m}\beta)^{2}\hat{N}\hat{L}(1\!+\!\eta)\big(\frac{3\mu}{(1\!-\!\eta)^3\eta}\\
&\!+\!\frac{\mu p}{(1\!-\!\eta)^2\eta^2}\big)\!+\!4m^2\bar{m}\hat{L}\hat{N}\beta(1\!+\!\eta)\big(\frac{\mu}{(1\!-\!\eta)^2\eta}\!+\!\frac{(p\!-\!1)\mu}{(1\!-\!\eta)\eta^{2}}\big)
\Big]\\	
&\!+\!\Big[4(m\!+\!1)\bar{m}^{2}\hat{L}\tilde{W}\beta\frac{(5\!-\!3\eta)}{(1\!-\!\eta)^{2}}\!+\!16m\bar{m}(m\!+\!\bar{m})\hat{L}\tilde{W}\beta\frac{(7\!-\!5\eta)}{(1\!-\!\eta)^{2}}\\
&\!+\!4\bar{m}^{2}\hat{L}\tilde{W}\beta^{2}\frac{(7\!-\!5\eta)(2\!-\!\eta)}{(1\!-\!\eta)^{4}}\!+\!8\bar{m}\hat{L}\tilde{W}(23m\!+\!3)\Big](t\!+\!1)\Big\}\\
&\!=\!\Big\{2\hat{N}\bar{L}m^3\beta\hat{N}(1\!+\!\eta)(\frac{\mu}{(1\!-\!\eta)^2}\!+\!\frac{(p-1)\mu}{(1\!-\!\eta)\eta})\!+\!2\hat{N}m\bar{m}^{2}\\
\end{aligned}
\end{equation*}
\begin{equation}
\begin{aligned}
&N_{f}\beta(1\!+\!\eta)(\frac{\mu}{(1\!-\!\eta)^2}\!+\!\frac{(p\!-\!1)\mu}{(1\!-\!\eta)\eta})+4m^3(\bar{m}\beta\hat{N})^{2}\hat{L}(1\!+\!\eta)\\	
&\big(\frac{3\mu}{(1\!-\!\eta)^3\eta}\!+\!\frac{\mu p}{(1\!-\!\eta)^2\eta^2}\big)\!+\!8m^3\bar{m}\hat{L}\hat{N}^{2}\beta(1\!+\!\eta)\big(\frac{\mu}{(1\!-\!\eta)^2\eta}\!+\!\\
&\frac{(p-1)\mu}{(1\!-\!\eta)\eta^{2}}\big)\Big\}+\Big\{2\hat{N}\bar{L}m(4m^{2}\tilde{W}\beta\frac{(5\!-\!3\eta)}{(1\!-\!\eta)^{2}}\!+\!24m\tilde{W})+\\
&8m(m\!+\!1)\bar{m}^{2}\hat{L}\hat{N}\tilde{W}\beta\frac{(5\!-\!3\eta)}{(1\!-\!\eta)^{2}}\!+\!32m^{2}\bar{m}(m\!+\!\bar{m})\hat{L}\hat{N}\tilde{W}\beta\\
&\frac{(7\!-\!5\eta)}{(1\!-\!\eta)^{2}}\!+\!8m\bar{m}^{2}\hat{L}\hat{N}\tilde{W}\beta^{2}\frac{(7\!-\!5\eta)(2\!-\!\eta)}{(1\!-\!\eta)^{4}}\!+\!16\hat{N}\bar{m}\hat{L}\\
&\tilde{W}m(23m\!+\!3)\Big\}(t+1)\triangleq q_{1}+q_{2}(t+1),
\label{eq:th27}
\end{aligned}
\end{equation}
where for the last equality, the coefficient of the term $t+1$ is defined as $q_{2}$, with the rest defined as $q_{1}$. Invoking the above result into \eqref{eq:th22} and considering $\epsilon_{0}\alpha_{0}\lambda_{0}=\epsilon_{-1}\alpha_{-1}=\frac{p(p-1)}{\mu(p-2)}$ yields
\begin{equation*}
	\begin{aligned}
		&\frac{\!(t\!+\!p\!+\!1)(t\!+\!p)}{2}\bar{D}_{\phi}(\mathbf{x}_{t+1},\mathbf{x}^{*})\!\\
		&\leq\epsilon_{0}\bar{D}_{\phi}(\mathbf{x}_{0},\mathbf{x}^{*})\!+\!\frac{4(q_{1}+q_{2}(t+1))}{\mu^{2}}+\frac{p(p-1)p_{\delta}}{\mu(p-2)}.
	\end{aligned}
\end{equation*}	
Then, the desired result follows directly. $\blacksquare$
\begin{rem}
Theorem \ref{th:thm2} reveals that Algorithm \ref{alg:alg1} achieves a convergence rate of $\mathcal{O}(1/k)$ under the selected step size and parameters. The design of the step size and parameters required to achieve faster convergence remains to be further investigated. Moreover, the upper bound on the number of iterations required to obtain an $\varepsilon$-approximate Nash equilibrium $\tilde{\mathbf{x}}$, i.e., $\bar{D}_{\phi}(\tilde{\mathbf{x}},\mathbf{x}^{*})\leq\varepsilon$ can be computed as $\mathcal{O}(\max\{\frac{L\sqrt{\bar{D}_{\phi}(\mathbf{x}_{0},\mathbf{x}^{*})}}{\mu\sqrt{\varepsilon}},\frac{\sqrt{q_{1}}}{\mu\sqrt{\varepsilon}},\frac{q_{2}}{\mu^{2}\varepsilon},\frac{\sqrt{Lp_{\delta}}}{\mu\sqrt{\varepsilon}}\})$. It can be observed from this bound that, with the other problem- and network-dependent quantities fixed, a smaller valid Lipschitz constant $L$ and a larger valid restricted strong monotonicity constant $\mu$ generally lead to a more favorable iteration-complexity bound and hence fewer iterations to attain a prescribed accuracy. This effect can be understood from the parameter design, since the constants $\mu$ and $L$ determine the step size $\alpha_k$ and the parameter $\lambda_k$, thereby affecting the convergence performance of the algorithm. 
\end{rem}

\section{Numerical Simulations}
\label{sec:NS}

In this section, the demand response management of energy systems is considered to illustrate the effectiveness of the proposed algorithm. It includes $h$ communities, and each community $l\in \lceil h \rceil$ consists of $m^{l}$ electricity users. For each user $j\in \mathcal{M}^{l}$, the electricity consumption over six operating periods is characterized by the strategy vector $x_{j}^{l}\in \Omega_{j}^{l}\subset\mathbb{R}^6$. The local constraint set is modeled as $\Omega_{j}^{l}=\{x_{j}^{l}\in\mathbb{R}^6: (x_{j}^{l}-\nu_{j}^{l})^{\top}Q(x_{j}^{l}-\nu_{j}^{l})\leq R^{2}\}$, where $\nu_{j}^{l}
\in\mathbb{R}^6$ denotes the nominal electricity-consumption profile and $R>0$ characterizes the overall load-adjustment capability, and $Q$ is a positive definite weighting matrix characterizing the adjustment flexibility over different operating periods. The cost function of each user $j$ is given by $f_{j}^{l}(x_{j}^{l},\psi(\mathbf{x}))=\lVert x_{j}^{l}-\nu_{j}^{l}\rVert^{2}+\left\langle U(\psi(\mathbf{x})),x_{j}^{l} \right\rangle$, where $\psi(\mathbf{x})=\sum_{l=1}^{h}\sum_{j=1}^{m^{l}}x_{j}^{l}$ is the aggregate formed by all users, and the unit price is defined as $U(\psi(\mathbf{x}))=a\psi(\mathbf{x})+\tilde{u}$ with $a>0$ and $\tilde{u}\in\mathbb{R}^{6}$ being a positive baseline price vector. In the simulation, consider $3$ communities containing $3$, $4$, and $5$ users, respectively. The interactions within and across communities are described by undirected time-varying communication graphs. For each communication network, the graph at each iteration is periodically selected from a set of three graphs, whose union is connected. Let $a=0.002$ and $\tilde{u}=[4,5,7,9,6,8]^{\top}$. For each user $j\in \mathcal{M}^{l}$, the components of the nominal electricity-consumption profile $\nu_{j}^{l}$ are randomly selected from $[40,80]$. The constraint set is specified by $R=15$ and a diagonal matrix $Q$ with diagonal entries $1,2,4,6,8$, and $10$. Additionally, the function $\phi$ is set as $\phi(x)=\frac{1}{2}x^{\top}Qx$, which yields the Bregman divergence $D_{\phi}(x_{1},x_{2})=\frac{1}{2}(x_{1}-x_{2})^{\top}Q(x_{1}-x_{2})$.

Algorithm \ref{alg:alg1} is implemented with $\alpha_{k}$ and $\lambda_{k}$ chosen according to Theorem \ref{th:thm2}. It is further compared with the Euclidean fully distributed algorithm in \citet{Zhao26} and the hierarchical Euclidean distributed algorithm with semi-decentralized aggregate acquisition in \citet{Chen24}. The parameters of the comparison algorithms are chosen following the settings in the corresponding references. Regarding the communication networks, the algorithm in \citet{Zhao26} employs the same networks as Algorithm \ref{alg:alg1}. For the algorithm in \citet{Chen24}, the same time-varying intra-cluster networks are adopted, while the cluster coordinators communicate over a time-varying network switching among three graphs whose union is connected. The simulation results presented in Fig. \ref{fig_2} show that the sequence $\mathbf{x}_{k}$ generated by Algorithm \ref{alg:alg1} converges to the NE $\mathbf{x}^{*}$. Moreover, Algorithm \ref{alg:alg1} achieves the same accuracy in less computation time than the other two algorithms. Therefore, compared with the Euclidean algorithms, the proposed framework can achieve improved convergence performance and computational efficiency by adapting the distance function to the geometry of the constraint set and incorporating gradient extrapolation based on historical gradient information.

\begin{figure}[!htp]
	\centering
	\includegraphics[width=3.1in]{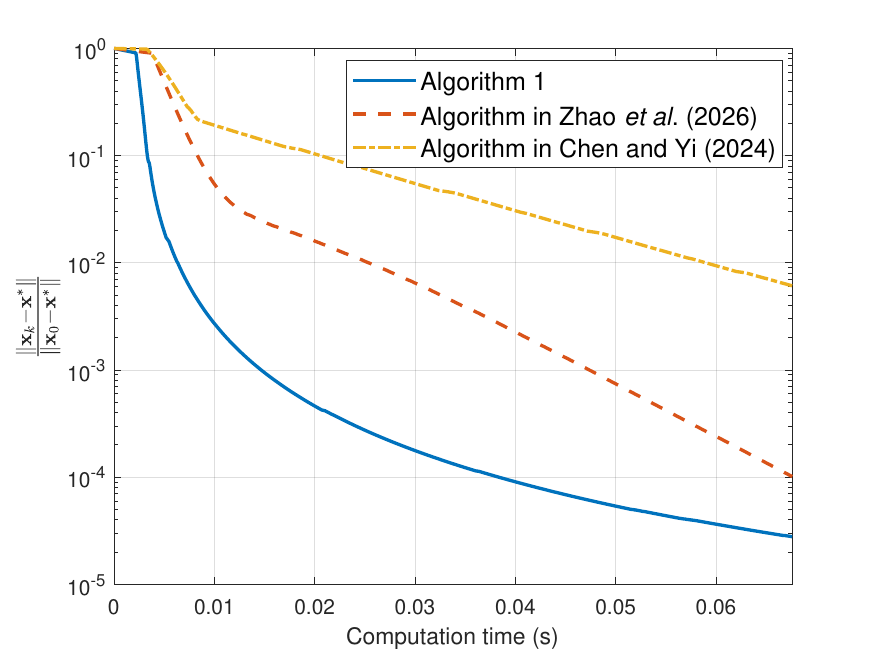}
	\caption{Comparison of the relative error $\frac{\lVert \mathbf{x}_{k}-\mathbf{x}^{*}\rVert}{\lVert \mathbf{x}_{0}-\mathbf{x}^{*}\rVert}$ among Algorithm \ref{alg:alg1} and the algorithms in \citet{Zhao26} and \citet{Chen24}. }
	\label{fig_2}
\end{figure}

\section{Conclusion}
\label{sec:Con}
This paper has designed a distributed algorithm in the non-Euclidean sense for MAGs, where competition and cooperation coexist. This algorithm has combined gradient extrapolation with mirror descent to ensure efficient convergence. It has been verified that this algorithm can converge at a rate of $\mathcal{O}(1/k)$ under the restricted strong monotonicity assumption. Moreover, the effectiveness of the proposed method has been demonstrated by the simulation results. Future work will focus on extending this algorithm to limited communication settings and relaxing the doubly stochastic assumption on adjacency matrices.






\end{document}